\documentclass[journal]{IEEEtran}
\usepackage{amsfonts}
\usepackage{dsfont}
\usepackage{amssymb}
\usepackage{graphicx}
\usepackage{subfigure}
\usepackage{enumerate}
\usepackage{amsmath}
\usepackage{color}
\usepackage{amsthm}
\usepackage{amsmath}
\usepackage{algorithm}
\usepackage{algpseudocode}
\usepackage[bookmarks=false]{hyperref}
\hypersetup{hidelinks}
\usepackage{breakurl}
\usepackage{cite}
\usepackage{epstopdf}
\usepackage{verbatim}

\allowdisplaybreaks[2]
\newcommand{\bp}{\begin{proof} \small }
\newcommand{\ep}{\end{proof} \normalsize}
\newcommand{\epx}{\end{proof} \small}
\newcommand{\bpa}{\begin{proofappx} \footnotesize }
\newcommand{\epa}{\end{proofappx} \small }
\newtheorem{theorem}{Theorem}

\newtheorem{remark}{Remark}
\newtheorem*{theorem*}{Theorem}
\newtheorem*{proposition*}{Proposition}
\newtheorem*{corollary*}{Corollary}
\newtheorem*{lemma*}{Lemma}
\newtheorem*{assumption*}{Assumption}
\newtheorem*{definition*}{Definition}
\newtheorem*{claim*}{Claim}

\newcommand{\be}{\begin{equation}}
\newcommand{\ee}{\end{equation}}
\newcommand{\bs}{\begin{subequations}}
\newcommand{\es}{\end{subequations}}
\newcommand{\bq}{\begin{eqnarray}}
\newcommand{\eq}{\end{eqnarray}}
\newcommand{\bqn}{\begin{eqnarray*}}
\newcommand{\eqn}{\end{eqnarray*}}

\newcommand{\ba}{\left[ \begin{array}}
\newcommand{\ea}{\\ \end{array} \right]}
\newcommand{\ben}{\begin{enumerate}}
\newcommand{\een}{\end{enumerate}}

\def\A{{\boldsymbol{A}}}

\def\H{{\boldsymbol{H}}}

\def\b{{\boldsymbol{b}}}

\def\k{{\boldsymbol{k}}}
\def\l{{\boldsymbol{l}}}

\def\w{{\boldsymbol{w}}}
\def\x{{\boldsymbol{x}}}

\def\z{{\boldsymbol{z}}}

\def\real{{\mathchoice%
{\hbox{\rm\setbox1=\hbox{I}\copy1\kern-.45\wd1 R}}
{\hbox{\rm\setbox1=\hbox{I}\copy1\kern-.45\wd1 R}}
{\hbox{\scriptsize\rm\setbox1=\hbox{I}\copy1\kern-.45\wd1 R}}
{\hbox{\scriptsize\rm\setbox1=\hbox{I}\copy1\kern-.45\wd1 R}}}}

\def\Zint{{\mathchoice{\setbox1=\hbox{\sf Z}\copy1\kern-.75\wd1\box1}
{\setbox1=\hbox{\sf Z}\copy1\kern-.75\wd1\box1}
{\setbox1=\hbox{\scriptsize\sf Z}\copy1\kern-.75\wd1\box1}
{\setbox1=\hbox{\scriptsize\sf Z}\copy1\kern-.75\wd1\box1}}}
\newcommand{\complex}{ \hbox{\rm C\kern-0.45em\rule[.07em]{.02em}{.58em}%
\kern 0.43em}}

\begin{document}

\title{Coded Computation across Shared Heterogeneous Workers with Communication Delay}
	
\author{ 
	Yuxuan Sun,~\IEEEmembership{Member,~IEEE,} 
	Fan Zhang, 
	Junlin Zhao,~\IEEEmembership{Member,~IEEE,} \\
	Sheng Zhou,~\IEEEmembership{Member,~IEEE,} 
	Zhisheng Niu,~\IEEEmembership{Fellow,~IEEE,} 
	Deniz G\"und\"uz,~\IEEEmembership{Senior Member,~IEEE}
	\thanks{ 
		Y. Sun, F. Zhang, S. Zhou and Z. Niu are with the Beijing National Research Center for Information Science and Technology, Department of Electronic Engineering, Tsinghua University, Beijing 100084, China (e-mail: sunyuxuan@tsinghua.edu.cn, zhang-f17@tsinghua.org.cn, sheng.zhou@tsinghua.edu.cn, niuzhs@tsinghua.edu.cn).}  
	\thanks{J. Zhao is with the Chinese University of Hong Kong (Shenzhen), Shenzhen 518172, China (e-mail: zhaojunlin@cuhk.edu.cn).}
	\thanks{D. G\"und\"uz is with the Department of Electrical and Electronic Engineering, Imperial College London, London SW7 2BT, UK (e-mail: d.gunduz@imperial.ac.uk).}
%	\thanks{Y. Sun, F. Zhang, S. Zhou and Z. Niu are sponsored in part by the National Key R\&D Program of China No. 2020YFB1806605, by the National Natural Science Foundation of China (No. 62022049, No. 61871254, No. 91638204, No. 61861136003), by the China Postdoctoral Science Foundation No. 2020M680558, and Hitachi Ltd. (Corresponding author: Sheng Zhou) }
%	\thanks{D. G\"und\"uz received funding from the European Research Council (ERC) under Starting Grant BEACON (grant No. 677854). } 
	\thanks{Part of this work has been presented in IEEE GLOBECOM 2019 \cite{Sun2019Globecom}. }
}

% make the title area
\maketitle

\begin{abstract}
	Distributed computing enables large-scale computation tasks to be processed over multiple workers in parallel. However, the randomness of communication and computation delays across workers causes the straggler effect, which may degrade the performance. Coded computation helps to mitigate the straggler effect, but the amount of redundant load and their assignment to the workers should be carefully optimized.
	In this work, we consider a \emph{multi-master heterogeneous-worker} distributed computing scenario, where multiple matrix multiplication tasks are encoded and allocated to workers for parallel computation. The goal is to minimize the communication plus computation delay of the slowest task. 
	We propose worker assignment, resource allocation and load allocation algorithms under both dedicated and fractional worker assignment policies, where each worker can process the encoded tasks of either a single master or multiple masters, respectively. Then, the non-convex delay minimization problem is solved by employing the Markov's inequality-based approximation, Karush-Kuhn-Tucker conditions, and successive convex approximation methods. Through extensive simulations, we show that the proposed algorithms can reduce the task completion delay compared to the benchmarks, and observe that dedicated and fractional worker assignment policies have different scopes of applications.
\end{abstract}


\IEEEpeerreviewmaketitle

\section{Introduction}	

With the fast development of artificial intelligence technologies and the explosion of data, computation tasks for the training and inference of machine learning (ML) models are becoming increasingly complex and demanding, which are almost impossible to be realized on a single machine. Distributed computing frameworks have been developed to parallelize these computations \cite{MapReduce,Googlelarge-scale}, where a centralized \emph{master} node takes charge of task partition, data dissemination, and result collection, and distributed computing nodes, called \emph{workers}, process partial computation tasks in parallel.

While parallel processing across multiple workers speeds up computation, the overall delay depends critically on the slowest worker. According to the experiments on the commercial Amazon elastic compute cloud (EC2) platform, some workers might experience much longer computation and communication delays than the average  \cite{lee2018speeding,codedhet,hierarchical}. This fact is mainly due to the randomness of the system, e.g., time-varying stochastic workloads of the workers, or the traffic over the communication network connecting the workers to the master. Such randomness leads to the so-called \emph{straggler effect}, which substantially increases the overall computation delay and becomes a major bottleneck in distributed computing.

The key idea to mitigate the straggler effect is to add redundancy to the computation tasks, so that the computation result does not rely on receiving results from all the workers. State-of-the-art approaches mainly include redundant scheduling of computation tasks \cite{gardner2017redundancy,joshi2017efficient,mma2019computation}, and various coding schemes \cite{Ng2020survey}, such as maximum distance separable (MDS) coding \cite{lee2018speeding,li2016unified,codedhet, hierarchical, Kim2021optimal, zhang2021coded}, gradient coding \cite{tandon2017,Reisizadeh2019tree,Bitar2020stochastic}, and polynomial coded computation \cite{Yu2019lagrange,Hasircioglu2021Bivariate}. Among them, the easiest policy is to replicate each task to multiple workers upon its arrival, and the optimal number of replicas can be derived under exponential \cite{gardner2017redundancy} or general service time distributions \cite{joshi2017efficient}. The orders of partitioned tasks at different workers are designed in \cite{mma2019computation}, and the impact of redundancy on the task completion delay under different scheduling orders is characterized.

Compared with simple task replication, coding can further improve the efficiency of computation. MDS coding schemes under different system settings have been widely investigated for matrix multiplication, which is the most common type of computation task in the distributed computing system. With $N$ homogeneous workers, it is proved in \cite{lee2018speeding} that MDS coding can reduce the computation delay by $O(\log N)$ compared to uncoded computation. Considering that workers have heterogeneous computing capabilities, the load allocation algorithms are proposed in \cite{codedhet} and \cite{Kim2021optimal} for a single-task scenario, both with asymptotic optimality. Based on \cite{codedhet}, an online, recursive load allocation algorithm is further proposed in \cite{zhang2021coded} for the random task arrival case, where cancellation is enabled to clear the unfinished parts of each task upon its completion, so as to avoid unnecessary computations. 

Although stragglers are slower than the average computation speed, it is still possible for them to provide partial results. This can be achieved by the hierarchical coded computation framework \cite{hierarchical}, or multi-message communications \cite{Ozfatura2019speeding, Ozfatura2020straggler, Buyukates2020timely, Hasircioglu2021Bivariate}. Specifically, in the hierarchical framework, the coded task at each worker is partitioned into multiple layers. Stragglers are able to finish the lower layer sub-tasks and thus the coding redundancy in the lower layers can be reduced to improve system efficiency \cite{hierarchical}.
Multiple messages that include partial computation result are allowed to transmit from each worker to the master at each time slot, and thus stragglers can contribute a few messages, not none, to the system \cite{Ozfatura2019speeding}. Multi-message communication may introduce additional transmission overhead, and the corresponding  trade-off of communication and computation delay is investigated in \cite{Ozfatura2020straggler}. 
Bivariate polynomial coding is introduced in \cite{Hasircioglu2021Bivariate}, and is shown to reduce the average computation delay with respect to univariate polynomial alternatives.
Such method is further combined with the concept of age of information for timely distributed computing in \cite{Buyukates2020timely}. 

The papers above mainly address the straggler effect caused by the randomness of computation delay. Meanwhile, as the communication data volume between the master and worker nodes is usually high, the communication delay cannot be ignored either. 
%In particular, with the development of edge computing and edge intelligence, 
Particularly, master and worker nodes might be base stations, mobile phones and smart vehicles at the edge of the wireless network, where the communication delay through wireless links may be highly stochastic and non-negligible. A scalable framework is proposed in \cite{Li2017scalable} for coded distributed computing over wireless networks, where the communication load does not scale with the number of workers. Considering an MDS-coded distributed computing system with homogeneous workers, the impact of packet erasure channel on the delay of tasks is analyzed in \cite{Han2019coded}. Under heterogeneous settings, fixed transmission rate is considered in \cite{Wu2020latency}, and the load allocation of MDS-coded tasks is optimized. A cooperative transmission scheme for coded matrix multiplication is proposed in \cite{Li2021coded} to reduce the inter-cell interference, while a joint coding and node scheduling algorithm is proposed in \cite{Huynh2021joint} based on reinforcement learning.

Most existing papers on distributed coded computing only consider a single-master scenario, and the impact of communication delay on the load allocation has not been sufficiently investigated. In this work, we consider a multi-master heterogeneous-worker distributed computing scenario, where multiple matrix multiplication tasks are encoded with MDS codes, and allocated to workers for parallel computation, with random communication and computation delay. 
The goal is to jointly design worker assignment and load allocation algorithms to minimize the completion delay of all the tasks. The main contributions of this work are summarized as follows:

1) We consider both dedicated and fractional worker assignment policies, where each worker can process the encoded tasks of either a single master or multiple masters, respectively. Considering the randomness of communication and computation delays, we formulate a unified delay minimization problem for the joint allocation of computing power, communication bandwidth and task load.

2) For dedicated worker assignment, we obtain a non-convex mixed-integer non-linear programming problem (MINLP). 
The load allocation problem is solved first by deriving a convex approximation problem with Markov's inequality.
%The optimal load allocations to a Markov's inequality based convex approximation problem, and a computation delay dominant special case optimization are derived, respectively. 
Worker assignment is then transformed to a max-min allocation problem, which is NP-hard and solved with greedy heuristics. A successive convex approximation (SCA) based algorithm is proposed to further enhance the load allocation.

3) For fractional worker assignment, the optimization problem is non-convex. We again use Markov's inequality to simplify the problem, and transform the fractional worker assignment and resource allocation problem to max-min allocation by deriving its optimality condition. A greedy algorithm is proposed accordingly.

4) Simulations under various settings verify the feasibility of the proposed Markov's inequality based approximation, and show the significant delay reduction of the proposed algorithms over benchmarks. In particular, when using Amazon EC2 for delay evaluation, about $82\%$ and $30\%$ delay reductions are achieved by the proposed algorithms compared to the uncoded and coded benchmarks, respectively. 

The rest of the paper is organized as follows. In Section \ref{sys}, we introduce the system model and formulate the problem. In Section \ref{sec_dedi}, we propose load allocation and worker assignment algorithms under the dedicated case. In Section \ref{sec_frac}, we further consider the fractional assignment case. Simulation results are shown in Section \ref{sim}, and conclusions are given in Section \ref{con}.

\begin{figure*}[!t]
	\centering
	\includegraphics[width=0.55\textwidth]{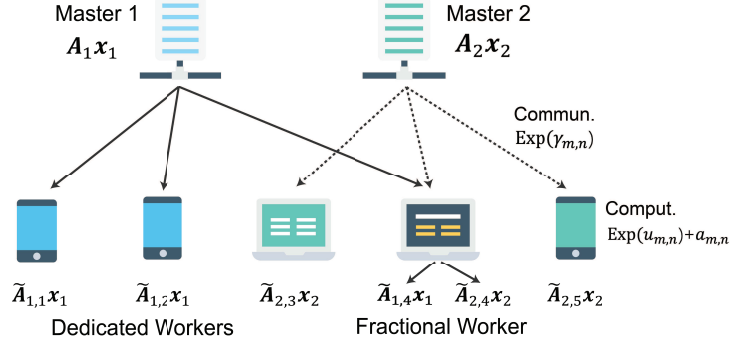}
	\vspace{-3mm}
	\caption{Illustration of a distributed computing system with multiple master and worker nodes.}	\label{system1}
	\vspace{-3mm}
\end{figure*}

\section{System Model and Problem Formulation} \label{sys}
As shown in Fig. \ref{system1}, we consider a distributed computing system with $M$ master nodes and $N$ worker nodes, denoted by  $\mathcal{M}=\{1,2,...,M\}$ and $\mathcal{N}=\{1,2,...,N\}$, respectively. Each master $m$ has a matrix-vector multiplication task, denoted by $\A_m\x_m$, where $\A_m\in \mathbb{R}^{L_m\times S_m}$, $\x_m\in \mathbb{R}^{S_m\times 1}$, and $L_m,S_m$ are positive integers. Each task can be partitioned and allocated to a subset of workers and processed by them in parallel. Local computation at the master is also available, and thus the set of nodes that can serve master $m$ is defined as $\mathcal{N}'\triangleq \mathcal{N}\cup\{0\}$, where index $0$ represents local processing.

To reduce the straggler effect brought by the randomness of communication and computation, we introduce redundancy to each task through MDS coded computation. Specifically, each master encodes matrix $\A_m$ in units of rows to get its coded version $\boldsymbol{\widetilde{A}}_m\in \mathbb{R}^{\widetilde{L}_m\times S_m}$, where $\widetilde{L}_m\geq L_m$ denotes the number of coded rows. 
Then, the coded matrix $\boldsymbol{\widetilde{A}}_m$ is divided into $N+1$ disjoint sub-matrices $\boldsymbol{\widetilde{A}}_{m,0}$, $\boldsymbol{\widetilde{A}}_{m,1}$, $\ldots$, $\boldsymbol{\widetilde{A}}_{m,N}$, where $\boldsymbol{\widetilde{A}}_{m,n}$ has $l_{m,n}\in\mathbb{N}$ rows, i.e., $\boldsymbol{\widetilde{A}}_{m,n}\in \mathbb{R}^{l_{m,n}\times S_m}$.
%Among which, $\boldsymbol{\widetilde{A}}_{m,0}\in \mathbb{R}^{l_{m,0}\times S_m}$ is processed by the master locally, where $l_{m,0}\in\mathbb{N}$ denotes the number of rows in $\boldsymbol{\widetilde{A}}_{m,0}$. For $\forall n\in\mathcal{N}$, define $l_{m,n}\in\mathbb{N}$ as the number of coded rows in $\boldsymbol{\widetilde{A}}_{m,n}$, i.e., $\boldsymbol{\widetilde{A}}_{m,n}\in \mathbb{R}^{l_{m,n}\times S_m}$, with $\sum_{n=0}^{N}l_{m,n}=\widetilde{L}_m$. 
Note that, $l_{m,n}=0$ indicates no assignment from master $m$ to worker $n$. Let $\Omega_m\triangleq\{n|n\in\mathcal{N}, l_{m,n}>0\}$ be the set of workers to serve master $m$. 

After task encoding and assignment, each master $m$ transmits $\boldsymbol{\widetilde{A}}_{m,n}$ and $\x_m$ to worker $n\in \Omega_m$ through their communication channel. We assume that the channel of each worker is orthogonal with that of others, and each worker can allocate its channel bandwidth to multiple masters and communicate with them simultaneously. This assumption is suitable for many realistic scenarios, e.g., the communication link is wired, or each worker is a base station with orthogonal wireless bandwidth.
%worker $n\in \Omega_m$ receives $\boldsymbol{\widetilde{A}}_{m,n}$ and $\x_m$ from master $m$, 
Each worker $n$ calculates the multiplication of $\boldsymbol{\widetilde{A}}_{m,n}$ and $\x_m$, and transmits back the result. 
Finally, master $m$ can recover the result of the original task $\A_m\x_m$ upon receiving the inner products of any $L_m$ out of $\widetilde{L}_m$ coded rows of $\boldsymbol{\widetilde{A}}_m$ and vector $\x_m$.

\subsection{Worker Assignment Policy}
We consider two worker assignment policies in this work:

1) \emph{Dedicated worker assignment:}
Each worker only serves a single master. 
For $\forall n\in\mathcal{N}, m\in\mathcal{M}$, let $k_{m,n}\in\{0,1\}$ be the worker assignment indicator, where $k_{m,n}=1$ if worker $n$ is assigned a coded task by master $m$, and $k_{m,n}=0$ otherwise. We have $\sum_{m=1}^{M}k_{m,n}\leq1, \forall n \in\mathcal{N}$.

2) \emph{Fractional worker assignment:}
We allow each worker to serve multiple masters simultaneously through processor sharing. Let $k_{m,n}\in [0,1]$ be the fraction of computing power of worker $n$ allocated to master $m$, with $\sum_{m=1}^{M}k_{m,n}\leq1, \forall n \in\mathcal{N}$. 
Define $b_{m,n}\in [0,1]$ as the fraction of bandwidth allocated to the link between master $m$ and worker $n$, with $\sum_{m=1}^{M}b_{m,n}\leq1, \forall n \in\mathcal{N}$.

We assume that a master is always dedicated, i.e., it only computes local task but not helping others. Therefore, for $\forall m\in\mathcal{M}$, we have $k_{m,0}=1$ and $b_{m,0}=1$.
Also note that, for dedicated worker assignment, the bandwidth allocation variable $b_{m,n}$ is binary, with $b_{m,n}=k_{m,n}$.

\subsection{Communication and Computation Delays}

We consider the delay of transmitting $\boldsymbol{\widetilde{A}}_{m,n}$ from master $m$ to worker $n$, and ignore the transmission delays of $\x_m$ and the computation results. This is because the size of $\boldsymbol{\widetilde{A}}_{m,n}$ is typically much larger than that of $\x_m$ and the result vector. Moreover, as $\x_m$ is shared among multiple workers that serve master $m$, $\x_m$ can be transmitted in a more efficient way, such as broadcast or multicast.

The communication delay to transmit a single coded row from master $m$ to worker $n$ using the whole bandwidth is modeled by an exponential distribution \cite{Han2019coded}, with rate parameter $\gamma_{m,n}$. Define the total communication delay of transmitting $\boldsymbol{\widetilde{A}}_{m,n}$ using $b_{m,n}$ of the bandwidth as $T^{\text{[tr]}}_{m,n}$, whose cumulative distribution function (CDF) is given by:
\begin{align}
&\mathbb{P}\left[T^{\text{[tr]}}_{m,n}(l_{m,n},b_{m,n};\gamma_{m,n})\leq t\right] \nonumber\\
&=
1-e^{-\frac{b_{m,n}\gamma_{m,n}}{l_{m,n}}t},~t\geq 0, \forall n\in\mathcal{N}, \forall m\in\mathcal{M}.
\end{align}
At each master, local processing does not need communication, and thus $T^{\text{[tr]}}_{m,0}=0, ~\forall m$.

Following the literature \cite{lee2018speeding,codedhet, hierarchical, zhang2021coded, Wu2020latency}, the delay of computing the inner product of one coded row of $\boldsymbol{\widetilde{A}}_m$ and vector $\x_m$ at worker $n$ or master $m$ ($n=0$) is modeled by a shifted exponential distribution, with shift parameter $a_{m,n}$ and rate parameter $u_{m,n}$. For $\forall n\in\mathcal{N}', ~\forall m\in\mathcal{M}$, define the total computation delay of $\boldsymbol{\widetilde{A}}_{m,n}\x_m$ as $T^{\text{[cp]}}_{m,n}$, with CDF
\begin{align} \label{Tcp_cdf}
\mathbb{P}\left[T^{\text{[cp]}}_{m,n}(l_{m,n},k_{m,n};a_{m,n},u_{m,n})\leq t\right]
~~~~~~~~~~~~~~~~~~ \nonumber\\
=
\begin{cases}
&1-e^{-\frac{k_{m,n}u_{m,n}}{l_{m,n}}\left(t-\frac{a_{m,n}l_{m,n}}{k_{m,n}}\right)},~t\geq \frac{a_{m,n}l_{m,n}}{k_{m,n}},\\
&0, ~~ \text{otherwise}.
\end{cases}
\end{align}

Let $T_{m,n}\triangleq T^{\text{[tr]}}_{m,n}+T^{\text{[cp]}}_{m,n}$ be the total communication plus computation delay of the task assigned from master $m$ to worker $n$, where $T^{\text{[tr]}}_{m,n}$ and $T^{\text{[cp]}}_{m,n}$ are two independent random variables. Then, if $b_{m,n}\gamma_{m,n}\neq k_{m,n}u_{m,n}$ and $t \geq \frac{a_{m,n}l_{m,n}}{k_{m,n}}$, the CDF of $T_{m,n}$ is given as follows:
\begin{flalign} \label{Tmn_cdf}
\mathbb{P} \! \left[T_{m,n}\leq t\right] 
&\!= \! 
1\! -\! 
\frac{b_{m,n}\gamma_{m,n}e^{-\frac{k_{m,n}u_{m,n}}{l_{m,n}}\left(t-\frac{a_{m,n}l_{m,n}}{k_{m,n}}\right)}}{b_{m,n}\gamma_{m,n}-k_{m,n}u_{m,n}} \nonumber\\
&+
 \frac{k_{m,n}u_{m,n}e^{-\frac{b_{m,n}\gamma_{m,n}}{l_{m,n}}\left(t-\frac{a_{m,n}l_{m,n}}{k_{m,n}}\right)}}{b_{m,n}\gamma_{m,n}-k_{m,n}u_{m,n}}.
\end{flalign}
%\begin{flalign} \label{Tmn_cdf}
%\mathbb{P} \! \left[T_{m,n}\leq t\right] \!\!= \!\! \nonumber\\
%\begin{cases}
%&\!\!\!\!\!\! 1\! -\! \frac{b_{m,n}\gamma_{m,n}e^{-\frac{k_{m,n}u_{m,n}}{l_{m,n}}\left(t-\frac{a_{m,n}l_{m,n}}{k_{m,n}}\right)}-k_{m,n}u_{m,n}e^{-\frac{b_{m,n}\gamma_{m,n}}{l_{m,n}}\left(t-\frac{a_{m,n}l_{m,n}}{k_{m,n}}\right)}}{b_{m,n}\gamma_{m,n}-k_{m,n}u_{m,n}}, 
%t\!\! \geq \!\!\frac{a_{m,n}l_{m,n}}{k_{m,n}},  \\
%&\!\!\!\!\!\! 0, ~~ \text{otherwise}.
%\end{cases}
%\end{flalign}
If $b_{m,n}\gamma_{m,n}=k_{m,n}u_{m,n}$ and $t\! \geq \!\frac{a_{m,n}l_{m,n}}{k_{m,n}}$, the CDF of $T_{m,n}$ is
\begin{flalign} \label{Tmn_cdf_eq}
&\mathbb{P}  \left[T_{m,n}\leq t\right] = 
 1 -  \nonumber\\
&\!\!\left[1\!\!+\!\!\frac{k_{m,n}u_{m,n}}{l_{m,n}}\!\!\left(\!t\!-\!\frac{a_{m,n}l_{m,n}}{k_{m,n}}\!\right)\!\right] 
\! e^{\!-\!\frac{k_{m,n}u_{m,n}}{l_{m,n}}\left(t-\frac{a_{m,n}l_{m,n}}{k_{m,n}}\right)}.
\end{flalign}
%\begin{flalign} \label{Tmn_cdf_eq}
%\mathbb{P}  \left[T_{m,n}\leq t\right] = 
%\begin{cases}
%&\!\!\!\!\!\! 1 -  \left[1+\frac{k_{m,n}u_{m,n}}{l_{m,n}}\left(t-\frac{a_{m,n}l_{m,n}}{k_{m,n}}\right)\right]  e^{-\frac{k_{m,n}u_{m,n}}{l_{m,n}}\left(t-\frac{a_{m,n}l_{m,n}}{k_{m,n}}\right)}, t\! \geq \!\frac{a_{m,n}l_{m,n}}{k_{m,n}},  \\
%&\!\!\!\!\!\! 0, ~~ \text{otherwise}.
%\end{cases}
%\end{flalign}
Otherwise, if $t<\frac{a_{m,n}l_{m,n}}{k_{m,n}}$, $\mathbb{P}  \left[T_{m,n}\leq t\right] = 0$.

For local computation, we have $T_{m,0}=T^{\text{[cp]}}_{m,0}, ~\forall m$. When $t\geq a_{m,0}l_{m,0}$, the CDF is given by 
\begin{align}\label{Tm0_cdf}
\mathbb{P}\left[T_{m,0}\leq t\right]=
1-e^{-\frac{u_{m,0}}{l_{m,0}}\left(t-a_{m,0}l_{m,0}\right)},~t\geq a_{m,0}l_{m,0},
\end{align}
otherwise, $\mathbb{P}\left[T_{m,0}\leq t\right]=0$.
%\begin{align}\label{Tm0_cdf}
%\mathbb{P}\left[T_{m,0}\leq t\right]=
%\begin{cases}
%&1-e^{-\frac{u_{m,0}}{l_{m,0}}\left(t-a_{m,0}l_{m,0}\right)},~t\geq a_{m,0}l_{m,0},\\
%&0, ~~ \text{otherwise}.
%\end{cases}
%\end{align}

\subsection{Problem Formulation}
Our objective is to minimize the task completion delay, by jointly optimizing the allocation of task load $\l\triangleq\{l_{m,n}|n\in\mathcal{N}', m\in\mathcal{M}\}$, computing power $\k\triangleq\{k_{m,n}|  n\in\mathcal{N}, m\in\mathcal{M}\}$, and communication bandwidth $\b\triangleq\{b_{m,n}|n\in\mathcal{N}, m\in\mathcal{M}\}$. As the communication and computation delays are with random, we aim to minimize the delay $t$, upon which the probability that all the masters can recover their computations is higher than a given threshold $\rho_s$. The optimization problem is formulated as:
\begin{subequations}
	\begin{align}
	\mathcal{P}1: \!\!\min_{\{\l,\k,\b,t\}} &~~~~t \label{ori_obj1} \\
	\text{s.t.} ~~&\mathbb{P}\left[X_m(t)\geq L_m\right]\geq \rho_s, ~ \forall m\in\mathcal{M},  \label{ori_cons_exp1} \\
	&\sum_{m=1}^{M} k_{m,n}\leq1, \sum_{m=1}^{M} b_{m,n}\leq1, \forall n\in\mathcal{N},   \label{ori_cons_sum1} \\
	& k_{m,n}\in \mathcal{K}, b_{m,n} \in \mathcal{K}, \forall m\in\mathcal{M},\forall n\in\mathcal{N}, \label{ori_cons_kb} \\
	& l_{m,n} \in \mathbb{N},~~\forall m\in\mathcal{M},~\forall n\in\mathcal{N}'. \label{ori_cons_l}
	\end{align}
\end{subequations}
In constraint \eqref{ori_cons_exp1}, $X_m(t)$ is defined as the number of computation results that can be received by the master $m$ until time $t$, where a unit result refers to the inner product of one coded row of $\boldsymbol{\widetilde{A}}_m$ and vector $\x_m$. Constraint \eqref{ori_cons_exp1} guarantees that each task can be recovered with probability $\rho_s$.
Equation \eqref{ori_cons_sum1} is the resource allocation constraint of each worker.
In constraint \eqref{ori_cons_kb}, we have $\mathcal{K}=\{0,1\}$ for dedicated worker assignment, while $\mathcal{K}=[0,1]$ for fractional worker assignment.
In constraint \eqref{ori_cons_l}, $\mathbb{N}$ represents the set of non-negative integers.

Since workers have heterogeneous computing and communication capabilities, their loads $l_{m,n}$ will be different in general. To derive $\mathbb{P}\left[X_m(t)\geq L_m\right]$, we need to find all the combinations of $\{l_{m,n},n\in\mathcal{N}'\}$ that satisfy $X_m(t)\geq L_m$, and further derive their joint probability distributions, which is intractable. As a result, problem $\mathcal{P}1$ is hard to solve.

We thus consider an approximation of $\mathcal{P}1$, where the probability constraint \eqref{ori_cons_exp1} is substituted by an expectation constraint, shown as follows:
\begin{subequations}
	\begin{align}
	\mathcal{P}2: \min_{\{\l,~\k,~\b,~t\}} &~~~~t \label{p1_obj1} \\
	\text{s.t.} ~~~~&\mathbb{E}[X_m(t)]\geq L_m, ~ \forall m\in\mathcal{M},  \label{cons_exp1} \\
	&l_{m,n} \geq 0,~~\forall m\in\mathcal{M},~\forall n\in\mathcal{N}, \label{ori_cons_lcont} \\
	&\text{Constraints } \eqref{ori_cons_sum1}, \eqref{ori_cons_kb}. \nonumber
	\end{align}
\end{subequations}
Constraint \eqref{cons_exp1} states that master $m$ is expected to receive sufficient computation results to recover $\A_m\x_m$ until time $t$. Similar approximation approach is also used in \cite{codedhet,zhang2021coded,Wu2020latency}, and the performance gap under a single master case can be bounded \cite{codedhet}. As $\A_m$ is with high dimension and thus the non-zero $l_{m,n}$ are typically large, we further relax $l_{m,n} \in \mathbb{N}$ to $l_{m,n} \geq 0$ in \eqref{ori_cons_lcont}, and ignore the rounding error in the following.

To simplify the system workflow as well as the theoretical analysis, we assume that each encoded task $\boldsymbol{\widetilde{A}}_{m,n}\x_m$, either being processed locally or allocated to a worker, is processed as a whole without any further partition. Accordingly, each master can only receive $l_{m,n}$ computation results from node $n\in\mathcal{N}'$ upon the completion. As computations on workers are independent, $\mathbb{E}[X_m(t)]$ can be written as follows:
\begin{align}\label{x_mn_def2}
\mathbb{E}[X_{m}(t)]=\sum_{n=0}^{N}\mathbb{E}\left[l_{m,n} \mathbb{I}_{\left\{T_{m,n}\leq t \right\}}\right] 
=\sum_{n=0}^{N}l_{m,n} \mathbb{P}\left[T_{m,n}\leq t\right], \nonumber
\end{align}
where $\mathbb{I}_{\left\{x\right\}}$ denotes the indicator function with $\mathbb{I}_{\left\{x\right\}}=1$ if event $x$ is true, and $\mathbb{I}_{\left\{x\right\}}=0$ otherwise.
For $n\in\mathcal{N}$, $\mathbb{P}\left[T_{m,n}\leq t\right]$ is given in \eqref{Tmn_cdf} or \eqref{Tmn_cdf_eq}, and for $n=0$, $\mathbb{P}\left[T_{m,n}\leq t\right]$ is given in \eqref{Tm0_cdf}.

%Since $\mathcal{P}1$ is not tractable, we will design its approximate solutions under dedicated and fractional worker assignments by solving $\mathcal{P}2$ in the following two sections. 
%In Section \ref{sim}, we will further test the delay performance of $\mathcal{P}1$ through simulations.

In the following two sections, we design solutions to $\mathcal{P}2$ under dedicated and fractional worker assignments, respectively. We will further show in Section \ref{sim} that a good solution to $\mathcal{P}2$ can also achieve low delay under the constraints of $\mathcal{P}1$.

\section{Dedicated Worker Assignment} \label{sec_dedi}

In this section, we solve problem $\mathcal{P}2$ under the dedicated worker assignment policy, where $\mathcal{K}=\{0,1\}$ and $b_{m,n}=k_{m,n}$. Accordingly, problem $\mathcal{P}2$ is a non-convex MINLP, which is very challenging to solve in general. 

We decouple the binary worker assignment variable $\k$ and the continuous load allocation variable $\l$ to seek a solution. 
First, given any worker assignment decision, the load allocation problem is still non-convex. We use Markov's inequality to provide a convex approximation to the non-convex constraint, and derive the optimal load allocation for this sub-problem. We also show that, when either the computation or communication delay plays a leading role, the original load allocation problem is convex, and the optimal solution can be derived.
Then, based on the optimal load allocation, we transform the worker assignment problem into a max-min allocation problem, which is still NP-hard and thus solved with greedy heuristics.
Finally, after optimizing the worker assignment, we further provide an enhanced load allocation algorithm by solving the original non-convex problem with the SCA method.

\subsection{Load Allocation for the General Case} \label{sub-load}

Given the set of workers $\Omega_m=\{n|k_{m,n}=1, n\in\mathcal{N}\}$ that serve master $m$, the optimal load allocation problem aims to minimize the task completion delay $t_m$ for each master $m$:
\begin{subequations}
	\begin{align}
	\mathcal{P}3: \min_{\{\l_m,~t_m\}} &~~~~t_m \\
	\text{s.t.} ~~&~\mathbb{E}[X_m(t_m)]\geq L_m,  \label{con_dedi}\\
	&~l_{m,n}\geq 0 , ~\forall n\in \Omega_m',
	\end{align}
\end{subequations}
where $\Omega_m'\triangleq \Omega_m\cup\{0\}$ includes the master $m$ itself, and $\l_m=\{l_{m,n}|n\in \Omega_m'\}$ denotes the load allocation vector. For $n\in\Omega_m$, the CDF of the total delay $T_{m,n}$ under dedicated assignment can be obtained by letting $k_{m,n}=1$ and $b_{m,n}=1$ in \eqref{Tmn_cdf} and \eqref{Tmn_cdf_eq}.
%\begin{align} \label{Tmn_cdf1}
%\mathbb{P}\left[T_{m,n}\leq t\right] = ~~~~~~~~~~~~~~~ \nonumber
%\end{align}
%\begin{align}
%\begin{cases}
%&\!\!\!\!\!\!1 \!- \!\frac{\gamma_{m,n}e^{-\frac{u_{m,n}}{l_{m,n}}\left(t-a_{m,n}l_{m,n}\right)}- u_{m,n}e^{-\frac{\gamma_{m,n}}{l_{m,n}}\left(t-a_{m,n}l_{m,n}\right)}}{\gamma_{m,n}-u_{m,n}}, \\ 
%&~~~~~~~~~~~ \gamma_{m,n}\neq u_{m,n}, t \!\geq\!a_{m,n}l_{m,n},\\
%&\!\!\!\!\!\! 1 \! - \! \left[1+\frac{u_{m,n}}{l_{m,n}}\left(t-a_{m,n}l_{m,n}\right)\right]  e^{-\frac{u_{m,n}}{l_{m,n}}\left(t-a_{m,n}l_{m,n}\right)}, \\
%&~~~~~~~~~~~ \gamma_{m,n}= u_{m,n}, t \!\geq\! a_{m,n}l_{m,n}, \\
%&\!\!\!\!\!\! 0, ~~ \text{otherwise}.
%\end{cases}
%\end{align}
Accordingly, $\mathbb{E}[X_m(t_m)]=\sum_{n\in \Omega_m'}l_{m,n}\mathbb{P}\left[T_{m,n}\leq t_m\right]$ is a non-convex function, making problem $\mathcal{P}3$ hard to solve. 

We provide an approximation to $\mathbb{E}[X_m(t_m)]$ based on Markov's inequality, i.e., for $n\in\Omega_m$,
\begin{align}
	\mathbb{P}\!\left[T_{m,n}\!\geq\! t_m \right] \!\leq \! \frac{E[T_{m,n}]}{t_m}
	\!=\!\frac{l_{m,n}}{t_m}\!\left(\!\frac{1}{\gamma_{m,n}} \!+\! \frac{1}{u_{m,n}}\!+\!a_{m,n}  \! \right) \!.
\end{align}
At the master, $\mathbb{P}\left[T_{m,0}\geq t_m\right]\leq \frac{l_{m,0}}{t_m}\left({\frac{1}{u_{m,0}}+a_{m,0} }\right)$.
Let
\begin{align} \label{def_theta}
	\theta_{m,n}\triangleq \frac{1}{\gamma_{m,n}}+\frac{1}{u_{m,n}}+a_{m,n}, ~\theta_{m,0}\triangleq\frac{1}{u_{m,0}}+a_{m,0}.
\end{align}
Then we have
\begin{align}\label{ineq_markov}
	\mathbb{E}[X_m(t_m)] &=\sum_{n\in \Omega_m'}l_{m,n}\mathbb{P}\left[T_{m,n}\leq t_m\right]  \nonumber\\
	&\geq \sum_{n\in\Omega_m'}l_{m,n}\left(1-\frac{\theta_{m,n} l_{m,n}}{t_m}\right).
\end{align}

Substituting inequality \eqref{ineq_markov} into \eqref{con_dedi}, we obtain a tighter constraint, and an approximation to $\mathcal{P}3$ is given by
\begin{subequations}
	\begin{align}
	\mathcal{P}4: \min_{\{\l_m,t_m\}} &~~~~t_m \\
	\text{s.t.} ~~~&~\sum_{n\in\Omega_m'}l_{m,n}\left(1-\frac{\theta_{m,n} l_{m,n}}{t_m}\right)\geq L_m,  \label{con_dedi_approx}\\
	&~l_{m,n}\geq 0, ~\forall n\in \Omega_m'. 
	\end{align}
\end{subequations}

Problem $\mathcal{P}4$ is a convex optimization problem, and the optimal solution is given as follows.

\begin{theorem} \label{markov_load}
	For a given subset of workers $\Omega_m$ that serves a master $m\in\mathcal{M}$, the optimal load allocation $\l_m^*$ and the corresponding task completion delay $t_m^*$ to $\mathcal{P}4$ are
	\begin{subequations}
		\begin{align}
		&l_{m,n}^*=\frac{L_m}{\theta_{m,n}\sum_{n\in\Omega_m'} \frac{1}{2\theta_{m,n}}},~n\in\Omega_m', \label{markov_l}  \\
		~~&t_m^*=\frac{L_m}{\sum_{n\in\Omega_m'} \frac{1}{4\theta_{m,n}}}. \label{markov_t}
		\end{align}
	\end{subequations}
\end{theorem}
\begin{proof}
	See Appendix \ref{a1}.
\end{proof}

As shown in \eqref{def_theta},  $\theta_{m,n}$ represents the expected total delay for worker $n$ to handle a unit coded task of master $m$, and thus $\frac{1}{\theta_{m,n}}$ indicates the average communication plus computation rate. As shown in Theorem 1, the optimal load allocated to each worker $n$ is proportional to $\frac{1}{\theta_{m,n}}$, while inversely proportional to the overall communication plus computation rates of workers. 

\subsection{Load Allocation for the Computation Delay Dominant Case} \label{sub-load-comp}

When computation delay is much larger than the communication delay, we ignore the latter and get $T_{m,n}=T^{\text{[cp]}}_{m,n}, \forall n\in\mathcal{N}'$. The CDF of $T_{m,n}$ is given in \eqref{Tcp_cdf}. It is easy to see that the optimal solution of $\mathcal{P}3$ must satisfy $t_m^*> \max_{\{n\in\Omega_m'\}}\{a_{m,n}l_{m,n}^*\}$.
In fact, if there is a worker $n_0\in\Omega_m'$ such that $t_m^*\leq a_{m,n_0}l_{m,n_0}$, then $l_{m,n_0}\mathbb{P}\left[T_{m,n_0}\leq t_m^*\right] =0$, meaning that the master $m$ cannot expect to obtain the computation results from worker $n_0$. By reducing $l_{m,n_0}$ to satisfy $t_m^*> a_{m,n_0}l_{m,n_0}$, constraint \eqref{con_dedi} can be strictly satisfied, and thus $t_m^*$ can be further reduced.

Based on this observation, constraint \eqref{con_dedi} of $\mathcal{P}3$ can be written as
\begin{align}
	\mathbb{E}[X_m(t_m)]&=\sum_{n\in \Omega_m'}l_{m,n}\mathbb{P}\left[T_{m,n}\leq t_m\right]  \nonumber\\
	&=\sum_{n\in \Omega_m'}l_{m,n} \left( 1-e^{-\frac{u_{m,n}}{l_{m,n}}\left(t_m-a_{m,n}l_{m,n}\right)} \right).
\end{align}

The following theorem provides the optimal solution to $\mathcal{P}3$.

\begin{theorem} \label{opt_load}
	When computation delay dominates the total delay, $\mathcal{P}3$ is a convex optimization problem, and the optimal load allocation $\l_m^*$ and task completion delay $t_m^*$ are
	\begin{subequations}
		\begin{align}
		&l_{m,n}^*=\frac{L_m}{\phi_{m,n}\sum_{n\in \Omega_m'} \frac{u_{m,n}}{1+u_{m,n}\phi_{m,n}}},~n\in\Omega_m',  \label{opt_l}\\
		~~&t_m^*= \frac{L_m}{\sum_{n\in \Omega_m'} \frac{u_{m,n}}{1+u_{m,n}\phi_{m,n}}}, \label{opt_t}
		\end{align}
	\end{subequations}
	where $\phi_{m,n}\triangleq \frac{1}{u_{m,n}}\left[-\mathcal{W}_{-1}(-e^{-u_{m,n}a_{m,n}-1} )-1\right]$, and $\mathcal{W}_{-1}(x)$ denotes the lower branch of Lambert W function, with $x\leq -1$ and $\mathcal{W}_{-1}(xe^x)=x$.
\end{theorem}
\begin{proof}
	See Appendix \ref{a2}.
\end{proof}

%[XXX Some further comments]  Comparing Theorem \ref{opt_load} with Theorem \ref{markov_load}, we can see that the expressions $l_{m,n}^*$ and $t_m^*$ have similar structures. Specifically, in \eqref{opt_l}

Similar results can be derived for the communication delay dominant case, by substituting $u_{m,n}$ with $\gamma_{m,n}$ and letting $a_{m,n}=0$.

\subsection{Dedicated Worker Assignment Algorithms} \label{sub-assign}

In this subsection, we design worker assignment algorithms, aiming to assign workers to masters in a balanced manner and minimize the completion delay of the slowest task.

According to Theorem \ref{markov_load}, the minimum task completion delay that can be achieved under a given subset of workers is
\begin{align}
	t_m^*=\frac{L_m}{\sum_{n\in\Omega_m'} \frac{1}{4\theta_{m,n}}}=\frac{L_m}{\frac{1}{4\theta_{m,0}}+\sum_{n=1}^{N} \frac{k_{m,n}}{4\theta_{m,n}}},
\end{align}
where we recall that $k_{m,n}\in\{0,1\}$ is the worker assignment indicator. 

From $\mathcal{P}2$, the objective of worker assignment is $\min_{\k} \max _{m\in\mathcal{M}}t_m^*$. As $t_m^*>0, \forall m$, the objective is equivalent to $\max_{\k} \min _{m\in\mathcal{M}}\frac{1}{t_m^*}$. Let $v_{m,n}\triangleq \frac{1}{4L_m\theta_{m,n}}, \forall m\in\mathcal{M},\forall n\in\mathcal{N}'$, and thus 
\begin{align} \label{frac_1_tm}
	\frac{1}{t_m^*}\!=\!\frac{1}{L_m}\!\left[\!\frac{1}{4\theta_{m,0}}\!+\!\sum_{n=1}^{N} \frac{k_{m,n}}{4\theta_{m,n}} \! \right]\!
	=\!v_{m,0} \!+ \!\sum_{n=1}^{N} k_{m,n}v_{m,n}.
\end{align}
The worker assignment problem can be transformed into the following form:
\begin{subequations}
	\begin{align}
	\mathcal{P}5: \min_{\k} &~\max _{m\in\mathcal{M}}~v_{m,0}+ \sum_{n=1}^{N} k_{m,n}v_{m,n} \label{obj_assign} \\
	\text{s.t.} &~~\sum_{m=1}^{M}k_{m,n}\leq 1, ~\forall n\in\mathcal{N}, \\
	&~~k_{m,n} \in \{0,1\} ,~~\forall m\in\mathcal{M},~\forall n\in\mathcal{N}. 
	\end{align}
\end{subequations}
Note that, for the computation delay dominant case, we only need to set $v_{m,n}=\frac{u_{m,n}}{L_m(1+u_{m,n}\phi_{m,n})}$, while the rest of the derivation still holds.

\begin{algorithm} [htb] 
	\caption{Iterated Greedy Algorithm for Dedicated Worker Assignment} \label{algo_iter}
	\begin{algorithmic}[1]
		\State \textbf{Input}: Let $v_{m,n}=\frac{1}{4L_m\theta_{m,n}}, \forall m,n$, $V_m=v_{m,0}, \forall m$ and $\Omega_m=\emptyset, \forall m$.
		\For {$n=1,...,N$}  \Comment{\textit{Initialization}}
		\State Find $m^*=\arg\max_{m\in \mathcal{M}} v_{m,n}$, and
		update $V_{m^*}=V_{m^*}+v_{m^*,n}$, $\Omega_{m^*}=\Omega_{m^*} \cup \{n\}$.
		\EndFor
		\While {termination condition is not satisfied}    \Comment{\textit{Main iteration}}
		\For {$n=1,... ,N$}    \Comment{\textit{Insertion}}
		\State Let $m_1$ be master that worker $n$ is serving, and $m_2=\arg \min_{m\in \mathcal{M}/\{m_1\}} V_{m}$.
		\State $V'_{m_1}=V_{m_1}-v_{m_1,n}$, $V'_{m_2}=V_{m_2}+v_{m_2,n}$, and $V'_{m}=V_{m}, \forall m \in\mathcal{M}/\{m_1,m_2\}$.
		\If {$\min_ {m\in \mathcal{M}}V'_{m}>\min_ {m\in \mathcal{M}}V_{m}$}
		\State $\Omega_{m_1}=\Omega_{m_1}-\{n\}$, $\Omega_{m_2}=\Omega_{m_2}+\{n\}$.
		\EndIf
		\EndFor
		\For {$n_1,n_2=1,... N$, and $n_1 \neq n_2$}    \Comment{\textit{Interchange}}
		\State Masters $m_1, m_2$ served by workers $n_1, n_2$, $V'_{m_1}=V_{m_1}-v_{m_1,n_1}+v_{m_1,n_2}$, and $V'_{m_2}=V_{m_2}-v_{m_2,n_2}+v_{m_2,n_1}$. 
		\If {$m_1 \neq m_2$, $v_{m_1,n_1}+v_{m_2,n_2}<v_{m_1,n_2}+v_{m_2,n_1}$, $V'_{m_1}>V_{\text{min}}$, and $V'_{m_2}>V_{\text{min}}$}
		\State $\Omega_{m_1}=\Omega_{m_1}-\{n_1\}+\{n_2\}$, $\Omega_{m_2}=\Omega_{m_2}-\{n_2\}+\{n_1\}$.
		\EndIf
		\EndFor			   
		\State Randomly remove some workers in $\mathcal{N}_s \in\mathcal{N}$, and update $V_m$ accordingly.  \Comment{\textit{Exploration}}
		\While {$\mathcal{N}_s \neq \emptyset$}
		\State Find $\{m^*,n^*\}=\arg\max_{m\in \mathcal{M}, n\in\mathcal{N}_s } v_{m,n}$.
		\State Update $V_{m^*}=V_{m^*}+v_{m^*,n^*}$, 
		$\Omega_{m^*}=\Omega_{m^*} \cup \{n^*\}$, $\mathcal{N}_s=\mathcal{N}_s-\{n^*\}$.
		\EndWhile
		\EndWhile
	\end{algorithmic}
\end{algorithm}

\begin{algorithm} [htb]
	\caption{Simple Greedy Algorithm for Dedicated Worker Assignment} \label{algo_simple}
	\begin{algorithmic}[1]
		\State \textbf{Input}: $\mathcal{N}_0=\{1,2,...,N\}$, $v_{m,n}=\frac{1}{4L_m\theta_{m,n}}, \forall m,n$, $V_m=v_{m,0}, \forall m$, and $\Omega_m=\emptyset, \forall m$.
		\While {$\mathcal{N}_0 \neq \emptyset$} 
		%		\State Solve the individual optimal worker assignment problem \textbf{P3}, and obtain $t_m^*$ for $m\in \mathcal{M}$.
		\State Find $m^*=\arg\min_{m\in \mathcal{M}}V_m$.
		\State Find $n^*=\arg\max_{ n\in\mathcal{N}_0 } v_{m^*,n}$.
		\State $V_{m^*}=V_{m^*}+v_{m^*,n^*}$.
		\State $\Omega_m=\Omega_m \cup \{n^*\}$, $\mathcal{N}_0=\mathcal{N}_0-\{n^*\}$.
		\EndWhile
	\end{algorithmic}
\end{algorithm}

Problem $\mathcal{P}5$ is called \emph{max-min allocation} problem, which is proposed for the fair assignment of items \cite{chakrabarty2009on,asadpour2010an}. In the original max-min allocation problem, each of the $N$ items has a unique value for an agent, and can be assigned to one of the $M$ agents. The objective is to assign all the items to the agents as fairly as possible, by maximizing the minimum total value of agents.
In $\mathcal{P}5$, each worker $n$ is an item with value $v_{m,n}$ for master $m$, and each master corresponds to an agent.
The max-min allocation problem can be reduced to the partitioning problem \cite{hayes2002}, when considering only 2 agents and assuming that each item has the same value for each agent. Since the partitioning problem is NP-complete, the max-min allocation problem is NP-hard.

Although some polynomial-time algorithms have been proposed for the max-min allocation problem with worst-case performance guarantee \cite{chakrabarty2009on,asadpour2010an}, they are very complex and difficult to implement. Instead, we propose two greedy algorithms in the following.

Inspired by \cite{peyro2010iterated}, an iterated greedy algorithm is proposed, as shown in Algorithm \ref{algo_iter}.
In the initialization phase, we assign each worker to the master with highest $v_{m,n}$, in order to maximize the contribution of workers.
Then, we iterate among the insertion, interchange, and exploration phases, until the termination condition is met. To be specific, in the insertion phase, each worker is re-assigned to a master $m_2$ with the minimum sum value $V_{m_2}$ if the minimum sum value of the masters is improved.
In the interchange phase, any two workers exchange the masters they are serving, if the minimum sum values of both masters, and the total value of the workers are improved. In the exploration phase, a subset of workers are randomly removed from the current assignment, and allocated to the masters in a greedy manner.
If the number of iterations reaches a preset maximum value, or the minimum sum value of the masters does not improve any more, the iteration is terminated. Note that, the final output is the worker assignment after the interchange phase.

As shown in Algorithm \ref{algo_simple}, we also propose a simple greedy algorithm that does not require iterations for performance improvement, inspired by the largest-value-first algorithm \cite{bryan1982scheduling}. The initial value of each master is related to its local computation capability, given by $V_m=v_{m,0} $. During the main loop, we select a master $m$ whose current sum value is the minimum, and allocate an available worker $n$ with highest $v_{m,n}$ for master $m$. The algorithm terminates when all the workers are allocated.

\subsection{SCA-Enhanced Load Allocation} \label{sub-SCA}

The main purpose of using Markov's inequality for load allocation in the general case is to provide an explicit form for the worker assignment problem. After that, we can get back to the original load allocation problem $\mathcal{P}3$ to further improve the performance. We observe that the non-convex constraint \eqref{con_dedi} in $\mathcal{P}3$ has a structure of the difference of convex functions, and thus we implement the SCA method to further optimize the load allocation.

%According to \eqref{Tmn_cdf1}, 
When $\gamma_{m,n}\neq u_{m,n}$, 

\begin{align}\label{EXt_dedi}
	&~\mathbb{E}[X_m(t_m)]=\sum_{n\in \Omega_m'}l_{m,n}\mathbb{P}\left[T_{m,n}\leq t_m\right] \nonumber\\
	&=l_{m,0} \left[1-e^{-\frac{u_{m,0}}{l_{m,0}}\left(t_m-a_{m,0}l_{m,0}\right)} \right] 
	+\sum_{n\in \Omega_m}l_{m,n}\Bigg[1-  \nonumber\\
	& \left. \frac{\gamma_{m,n}e^{-\frac{u_{m,n}}{l_{m,n}}\left(t-a_{m,n}l_{m,n}\right)} \!\!- \!\! u_{m,n}e^{-\frac{\gamma_{m,n}}{l_{m,n}}\left(t-a_{m,n}l_{m,n}\right)}}{\gamma_{m,n}-u_{m,n}}\right].
\end{align} 

Let $\w_m\triangleq \{\l_m,t_m\}$. Without loss of generality, we assume $\gamma_{m,n}>u_{m,n}$, and let 
\begin{align}
	&h_{m,n}^+(\w_m)\triangleq \frac{\gamma_{m,n}l_{m,n} e^{-\frac{u_{m,n}}{l_{m,n}}\left(t-a_{m,n}l_{m,n}\right)}}{\gamma_{m,n}-u_{m,n}}, \nonumber\\
	&h_{m,n}^-(\w_m) \triangleq \frac{u_{m,n}l_{m,n}e^{-\frac{\gamma_{m,n}}{l_{m,n}}\left(t-a_{m,n}l_{m,n}\right)}}{\gamma_{m,n}-u_{m,n}} . \nonumber
\end{align}
%$h_{m,n}^+(\w_m)\triangleq \frac{\gamma_{m,n}l_{m,n} e^{-\frac{u_{m,n}}{l_{m,n}}\left(t-a_{m,n}l_{m,n}\right)}}{\gamma_{m,n}-u_{m,n}}$,
%$h_{m,n}^-(\w_m) \triangleq \frac{u_{m,n}l_{m,n}e^{-\frac{\gamma_{m,n}}{l_{m,n}}\left(t-a_{m,n}l_{m,n}\right)}}{\gamma_{m,n}-u_{m,n}}$.
Otherwise, we can exchange $h_{m,n}^+(\w_m)$ with $h_{m,n}^-(\w_m)$, and the following solution still works. Let $h_{m,0}(\w_m)\triangleq -l_{m,0} \left[1-e^{-\frac{u_{m,0}}{l_{m,0}} \left(t_m-a_{m,0}l_{m,0}\right)} \right]$. From Appendix \ref{a2}, we know that $h_{m,0}(\w_m)$, $h_{m,n}^+(\w_m)$, and $h_{m,n}^-(\w_m)$ are all convex functions. Accordingly,
\begin{align}
	L_m-\mathbb{E}[X_m(t_m)]=L_m-\sum_{n\in \Omega_m}l_{m,n}+h_{m,0}(\w_m)& \nonumber\\
	+\sum_{n\in \Omega_m}\left(h_{m,n}^+(\w_m)-h_{m,n}^-(\w_m)\right),&
\end{align}
that is, $L_m-\mathbb{E}[X_m(t_m)]$ can be decomposed into the difference of convex functions.

For any given point $\z$, a convex upper bound of $h_{m,n}^+(\w)-h_{m,n}^-(\w)$ can be obtained by linearizing $h_{m,n}^-(\w)$:
\begin{align}
	&h_{m,n}^+(\w)-h_{m,n}^-(\w) \nonumber\\
	&\leq h_{m,n}^+(\w)-h_{m,n}^-(\z)-\nabla_{\w}h_{m,n}^-(\z)^{T}(\w-\z).
\end{align} 
Let $\tilde{h}_{m,n}(\w,\z) \triangleq h_{m,n}^+(\w)-h_{m,n}^-(\z)-\nabla_{\w}h_{m,n}^-(\z)^{T}(\w-\z)$. A convex approximation problem to $\mathcal{P}3$ under point $\z$, denoted by $\mathcal{P}(\z)$, is given by 
\begin{subequations}
	\begin{align}
	\mathcal{P}\!(\z): \!\min_{\{\w_m\}} &~~t_m \\
	\text{s.t.} ~~&  h_{m,0}(\w_m) \!+\!\!\!\sum_{n\in \Omega_m}\!\!\tilde{h}_{m,n}(\w_m,\z)\!\leq\!\!\!\! \sum_{n\in \Omega_m}\!\!\! l_{m,n}\!\! -\!\!L_m ,  \label{con_dedi_sca}\\
	&l_{m,n}\geq 0 , ~\forall n\in \Omega_m'.
	\end{align}
\end{subequations}

\begin{algorithm} [!t]
	\caption{SCA-Enhanced Load Allocation Algorithm} \label{algo_SCA}
	\begin{algorithmic}[1]
		\State \textbf{Input}: Given master $m$ and its worker assignment $\Omega_m$, find a feasible point $\z_0$ of $\mathcal{P}3$, and set $\gamma_0=1$, $r=0$, $\alpha \in (0,1)$.
		\While {$\z_r$ is not a stationary solution} 
		\State Solve the optimal solution $\w_r$ to the optimization problem $\mathcal{P}(\z_r)$.
		\State $\z_{r+1}=\z_r+ \gamma_r (\w_r-\z_r)$.
		\State $\gamma_{r+1}=\gamma_r(1-\alpha \gamma_r)$, $r \leftarrow r+1$.
		\EndWhile
	\end{algorithmic}
\end{algorithm}

Based on the SCA method proposed in \cite{scutari2017p1}, we develop an SCA-enhanced load allocation algorithm, as shown in Algorithm \ref{algo_SCA}.
For each master $m$ and the corresponding worker assignment $\Omega_m$ by Algorithm \ref{algo_iter} or \ref{algo_simple}, the SCA algorithm starts from a feasible point of $\z_0$ of $\mathcal{P}3$. Note that, the Markov's inequality provides a tighter approximation to constraint \eqref{con_dedi}, and thus Theorem \ref{markov_load} directly provides $\z_0$. Then, we iteratively solve convex optimization problems $\mathcal{P}(\z_r)$ until convergence, where in the $r$-th iteration, $\z_r$ is updated according to Line 4 using step-size $\gamma_r$. According to \cite{scutari2017p1}, we update $\gamma_r$ with a decreasing ratio $\alpha\in(0,1)$, so as to guarantee the convergence to a local optimum.

As a summary, we would like to provide the following remarks.

\begin{remark}
	Scope of application:
	\emph{While we assumed certain delay distributions in the system model, the Markov's inequality based approximate load allocation and the corresponding worker assignment algorithms, introduced in Section \ref{sub-load} and Section \ref{sub-assign}, \emph{do not} rely on these distributions. Instead, the proposed solution can be applied to any communication and computation delay distributions with broad adaptivity, as long as their mean values are known.
	To further carry out the SCA-enhanced load allocation, we need to specify the delay distributions.}
\end{remark}

\begin{remark}
	Iterated matrix multiplication:
	\emph{Distributed matrix-vector multiplication is often needed for the training of large ML models, where matrix $\A_m$ corresponds to the data and vector $\x_m$ to the model \cite{lee2018speeding,Ozfatura2019speeding}. Using a common training algorithm such as distributed gradient descent, the coded data is transmitted to the workers at the beginning, while multiple iterations of computations are required with the updated model vector. In this scenario, we can use the result of the computation-delay dominant case for worker assignment and load allocation, or modify the communication delay distribution of $\x_m$ by removing the load variable $l_{m,n}$. }
%		Moreover, as the task completion delays over iterations are independent, we focus on a one-shot optimization problem in this work.}
\end{remark}

\section{Fractional Worker Assignment} \label{sec_frac}
While dedicated worker assignment only needs a simple communication connection topology between masters and workers, it may lead to an unbalanced worker assignment, particularly when a few workers are much more powerful than the others, or the number of workers is relatively small. Therefore, in this section, we further consider fractional worker assignment, by allowing each worker to serve multiple masters simultaneously. In this case, we have $\mathcal{K}=[0,1]$, $k_{m,n},b_{m,n}\in\mathcal{K}$, $\forall m,n$, and the CDF of the total delay $T_{m,n}$ is given in \eqref{Tmn_cdf} and \eqref{Tmn_cdf_eq}. Accordingly, problem $\mathcal{P}2$ is a non-convex optimization problem, which is difficult to solve directly.

Similarly to Section \ref{sec_dedi}, we use Markov's inequality to derive an approximation to problem $\mathcal{P}2$, and further simplify the resultant optimization problem by analyzing its optimality condition. We show that the joint bandwidth and computing power allocation under fractional assignment can also be transformed to a max-min allocation problem, and propose a greedy algorithm based on Algorithms \ref{algo_iter} and \ref{algo_simple}.

\subsection{Markov's Inequality based Approximation and its Optimality Condition}
Using the Markov's inequality, $\forall m\in\mathcal{M}$ and $\forall n \in\mathcal{N}$ with $b_{m,n}\neq 0$ and $k_{m,n}\neq 0$, 
\begin{align}\label{approx_T_frac}
&\mathbb{P}\left[T_{m,n}\leq t\right]=1-\mathbb{P}\left[T_{m,n}\geq t\right] \geq 1- \frac{E[T_{m,n}]}{t} \nonumber\\
&=1-\frac{l_{m,n}}{t} \left(\frac{1}{b_{m,n}\gamma_{m,n}}+\frac{1}{k_{m,n}u_{m,n}}+\frac{a_{m,n}}{k_{m,n}}   \right).
\end{align}

In the fractional assignment case, the expected total delay for worker $n\in\mathcal{N}$ to handle a unit coded task of master $m\in\mathcal{M}$ is given by
\begin{align} \label{theta_mn_kb}
	\theta_{m,n}\!\!=\!\!
%	\frac{1}{b_{m,n}\gamma_{m,n}}+\frac{1}{k_{m,n}u_{m,n}}+\frac{a_{m,n}}{k_{m,n}}.
	\begin{cases}
	&\!\!\!\!\!\!\frac{1}{b_{m,n}\gamma_{m,n}}\!+\!\frac{1}{k_{m,n}u_{m,n}}\!+\!\frac{a_{m,n}}{k_{m,n}}, ~k_{m,n}, b_{m,n}\!>\!0,\\
	&\!\!\!\!\!\!\infty, ~~ k_{m,n}=0 \text{~or~} b_{m,n}=0.
	\end{cases}
\end{align}
For local computation at each master $m$, we still have $\theta_{m,0}=\frac{1}{u_{m,0}}+a_{m,0}$. Considering the inherent feature of the system, $k_{m,n}$, $b_{m,n}$ and $l_{m,n}$ are either all non-zero or all zero.
%,and $\lim_{k_{m,n}\to 0 }\theta_{m,n}l_{m,n}=0$, $\lim_{b_{m,n}\to 0 }\theta_{m,n}l_{m,n}=0$.

Substituting \eqref{approx_T_frac} into \eqref{con_dedi}, an approximation to problem $\mathcal{P}2$ under the fractional worker assignment policy is given by
\begin{subequations}
\begin{align}
	\mathcal{P}6\!: \!\! \min_{\{\l,~\k,~\b,~t\}} &~~~~t \\
	\text{s.t.} ~~~~& L_m \!\! -\!\!\sum_{n=0}^{N}\!l_{m,n} \left(1\!-\!\frac{l_{m,n}\theta_{m,n}}{t}\! \right)\!\!\leq \!0
	, \forall m, \label{cons_frac_Lm_approx}\\
	&\sum_{m=1}^{M} k_{m,n}\leq1, ~~\sum_{m=1}^{M} b_{m,n}\leq1,  ~~ \forall n,   \label{cons_frac_kb} \\
	& k_{m,n},b_{m,n} \in [0,1], l_{m,n} \geq 0,~~\forall m, n. \label{cons_frac_l}
\end{align}
\end{subequations}
%In \eqref{cons_frac_Lm_approx}, $\frac{l_{m,n}^2}{b_{m,n}}$, $\frac{l_{m,n}^2}{k_{m,n}}$ and $\frac{1}{t}$ are convex,  which can be written as the difference of convex functions, and thus $\mathcal{P}7$ can be solved by SCA. But this is too complex and impractical.

Compared to $\mathcal{P}4$, problem $\mathcal{P}6$ needs to jointly optimize load allocation $\l$ and the resource allocation $\k$ and $\b$, which is still non-convex. In the following theorem, we derive the KKT optimality condition for $\mathcal{P}6$.
\begin{theorem} \label{opt_cond}
	Given any resource allocation $k_{m,n}$ and $b_{m,n}$, the optimal load allocation $l_{m,n}^*$ to problem $\mathcal{P}6$ that minimizes delay $t^*$ must satisfy the following condition:
		\begin{align}
		&l_{m,n}^*=\frac{t^*}{2\theta_{m,n}},~n\in\mathcal{N}',  \label{opt_cond_l}
		\end{align}
		where $\theta_{m,n}$ is derived from \eqref{theta_mn_kb} according to $k_{m,n}$ and $b_{m,n}$.
\end{theorem}
\begin{proof}
	See Appendix \ref{a3}.
\end{proof}

\subsection{Fractional Worker Assignment Algorithm}
\begin{algorithm} [!t]
	\caption{Greedy Algorithm for Fractional Worker Assignment} \label{algo_frac}
	\begin{algorithmic}[1]
		\State \textbf{Input}: Get an initial dedicated worker assignment according to Algorithm \ref{algo_iter} or \ref{algo_simple}. Let $b_{m,n}=k_{m,n}, \forall m,n$. Initialize $\theta_{m,n}$ according to \eqref{theta_mn_kb}, and let $V_m= \frac{1}{L_m} \sum_{n=0}^{N} \frac{1}{4\theta_{m,n}}, \forall m$. 
		\While {$\max_m V_m > \min_m V_m$ }
		\State $m_1=\arg\max V_m $, $m_2=\arg\min V_m $, $\mathcal{N}_{\text{tmp}}=\{n|k_{m_1,n}>0, k_{m_2,n}=0\}$.
		\State Calculate $\theta_{m_2,n}'=\frac{1}{b_{m_1,n}\gamma_{m_2,n}}+\frac{1}{k_{m_1,n}u_{m_2,n}}+\frac{a_{m_2,n}}{k_{m_1,n}}$, $\forall n\in\mathcal{N}_{\text{tmp}}$.
%		\While {$\{n|k_{m_1,n}>0, k_{m_2,n}=0\}\neq \emptyset$ and $V_{m_1}> V_{m_2}$}
		\State Find worker $n_1=\arg\max_{n\in\mathcal{N}_{\text{tmp}}}\frac{1}{\theta_{m_2,n}'}$ with maximum performance gain for master $m_2$.
%		, i.e., $n_1=\arg\max_{n\in\mathcal{N}_{\text{tmp}}}\frac{1}{\theta_{m_2,n}'}$.
%		\begin{align}
%			n_1=\arg\max_{n\in\mathcal{N}_{\text{tmp}}}\frac{1}{\theta_{m_2,n}'}.
%		\end{align} 
		\If {$V_{m_1}-\frac{1}{4\theta_{m_1,n_1}L_{m_1}}\leq V_{m_2}+\frac{1}{4\theta_{m_2,n_1}'L_{m_2}}$}
		\State Find $k_{m_1,n_1}, b_{m_1,n_1}, k_{m_2,n_1}, b_{m_2,n_1}$, such that $V_{m_1}=V_{m_2}$.
		\Else 
		\State Assign all the resource of worker $n_1$ for master $m_1$ to master $m_2$, i.e., let $k_{m_2,n_1}=k_{m_1,n_1}$, $b_{m_2,n_1}=b_{m_1,n_1}$, and then let $k_{m_1,n_1}=0$, $b_{m_1,n_1}=0$.
		\EndIf
		\State Update $\theta_{m_1,n_1}$,  $\theta_{m_2,n_1}$, $V_{m_1}$, $V_{m_2}$ accordingly.
%		\EndWhile
		\EndWhile
	\end{algorithmic}
\end{algorithm}

Based on Theorem \ref{opt_cond}, without loss of optimality, constraint \eqref{cons_frac_Lm_approx} can be simplified to
\begin{align}
L_m- \sum_{n=0}^{N} \frac{t}{4\theta_{m,n}}\leq 0,~\forall m.
\end{align}
Therefore, problem $\mathcal{P}6$ is equivalent to:
\begin{subequations}
	\begin{align}
	\mathcal{P}7:  \max_{\{\k,~\b\}} &\min_{m\in\mathcal{M}} \frac{1}{L_m} \sum_{n=0}^{N} \frac{1}{4\theta_{m,n}}\\
	\text{s.t.} ~&~~\text{constraints~}\eqref{cons_frac_kb}, \eqref{cons_frac_l}. 
	\end{align}
\end{subequations}
We can see that $\mathcal{P}7$ is very similar to the max-min allocation problem $\mathcal{P}5$ under dedicated worker assignment, except that $\theta_{m,n}$ can further change with respect to the computing power allocation $k_{m,n}\in[0,1]$ and communication bandwidth allocation $b_{m,n}\in[0,1]$. Therefore, we adopt the dedicated assignment as an initialization, and iteratively balance the resource allocation between the master $m_1=\arg\max V_m$ with maximum sum value and  master $m_2=\arg\min V_m$ with minimum sum value. To balance their sum values, we select a worker $n_1$ that serves $m_1$ but not $m_2$ for the moment, with maximum potential performance gain for $m_2$, as shown in Lines 3-5. Then, part or all of the computing power and communication bandwidth of worker $n_1$ are re-assigned to master $m_2$, as shown in Lines 6-10.
Note that, in practice, we may not want to make the topology of masters and workers too complicated. In that case, we can limit the maximum number of masters each worker can serve in Algorithm \ref{algo_frac}.

We also remark that, by substituting $\gamma_{m,n}$, $u_{m,n}$ and $a_{m,n}$ in \eqref{EXt_dedi} with $b_{m,n}\gamma_{m,n}$, $k_{m,n}u_{m,n}$ and $\frac{a_{m,n}}{k_{m,n}}$, respectively, the SCA-enhanced load allocation (Algorithm \ref{algo_SCA}) can also be implemented for the fractional worker assignment problem after Algorithm \ref{algo_frac}.

\section{Simulation Results} \label{sim}

In this section, we show the simulation results of the proposed algorithms under various settings. We first verify the feasibility of the Markov's inequality based approximation, and then evaluate the task completion delay of the proposed algorithms and benchmarks. Finally, we sample the task completion delay from commercial compute platform Amazon EC2, and use the measured data to further validate the proposed algorithms. 

\subsection{Validation of Markov's Inequality based Approximation} \label{sim-vali}

\begin{figure}[t]
	\centering
	\subfigcapskip=-1mm
	\subfigure[Average task completion delay.]{\label{fig_val_m2n5_bar}			
		\includegraphics[width=0.42\textwidth]{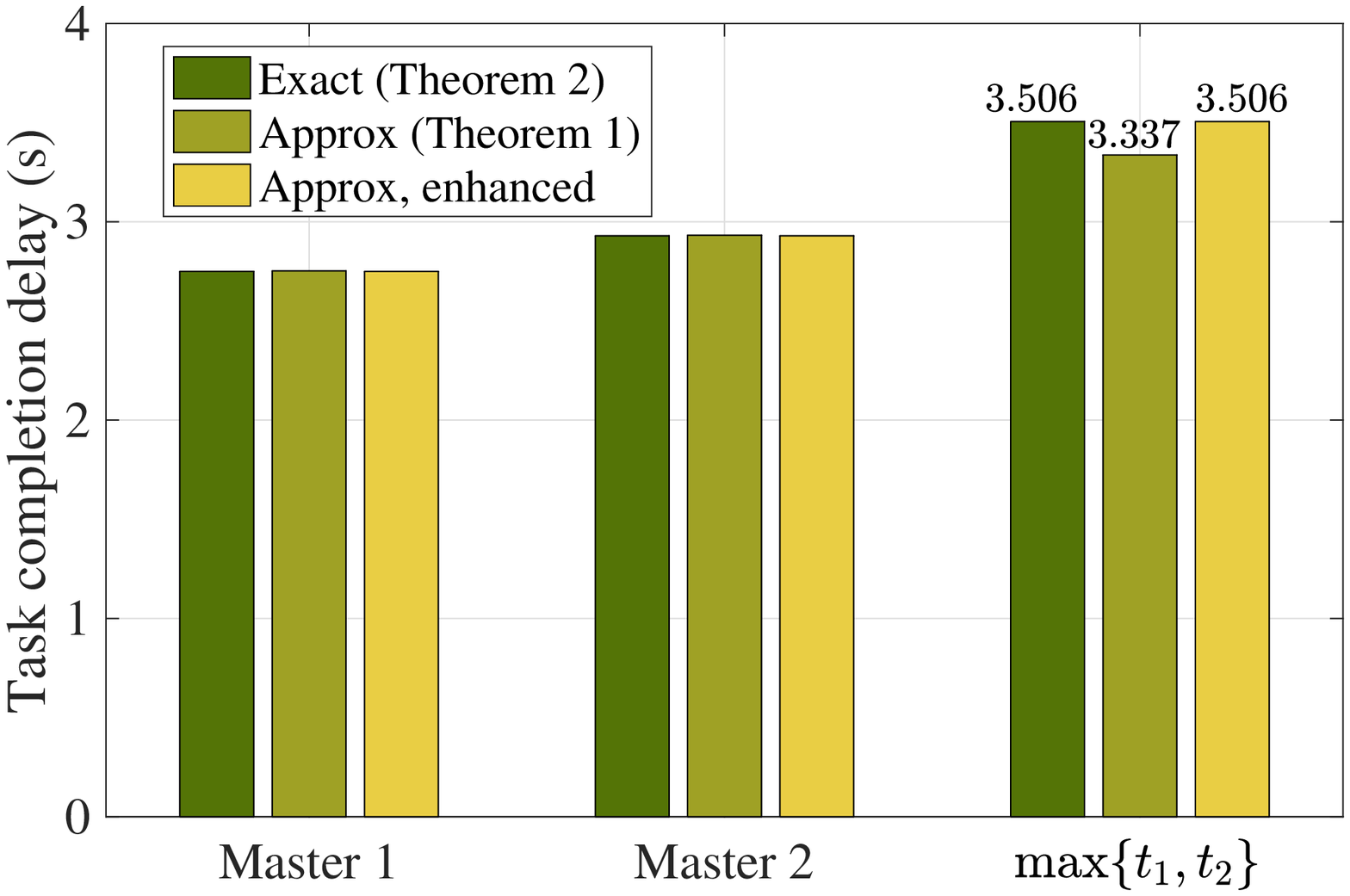}}
	\subfigure[CDF of task completion delay.]{\label{fig_val_m2n5_cdf}			
		\includegraphics[width=0.42\textwidth]{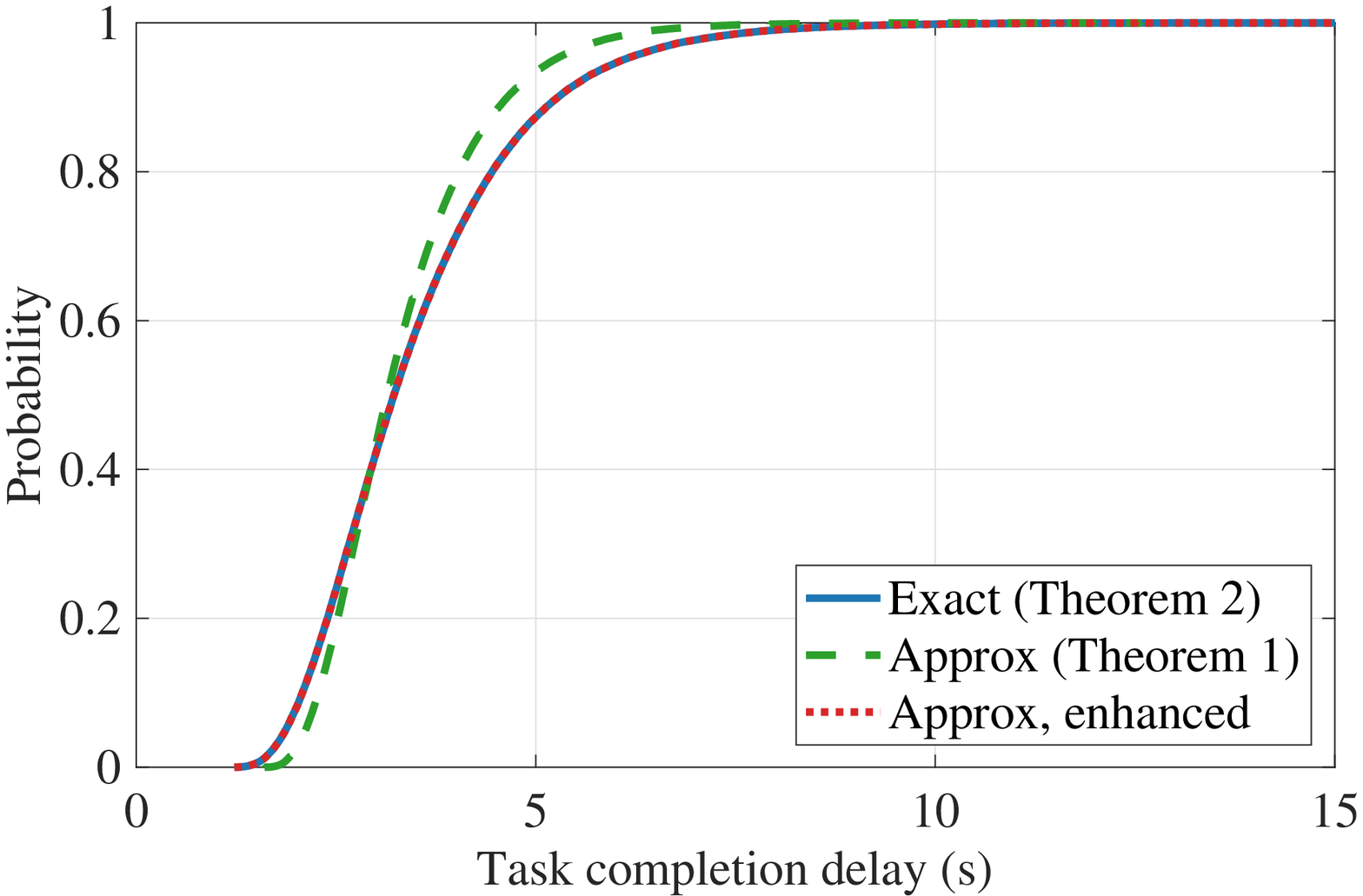}}
	\caption{Validation of the Markov's inequality based approximation in the $2$-master, $5$-worker case.}	
	\vspace{-1mm}\label{fig_val_m2n5}
\end{figure}

\begin{figure}[t]
	\centering
	\subfigcapskip=-1mm
	\subfigure[Average task completion delay.]{\label{fig_val_m4n50_bar}			
		\includegraphics[width=0.42\textwidth]{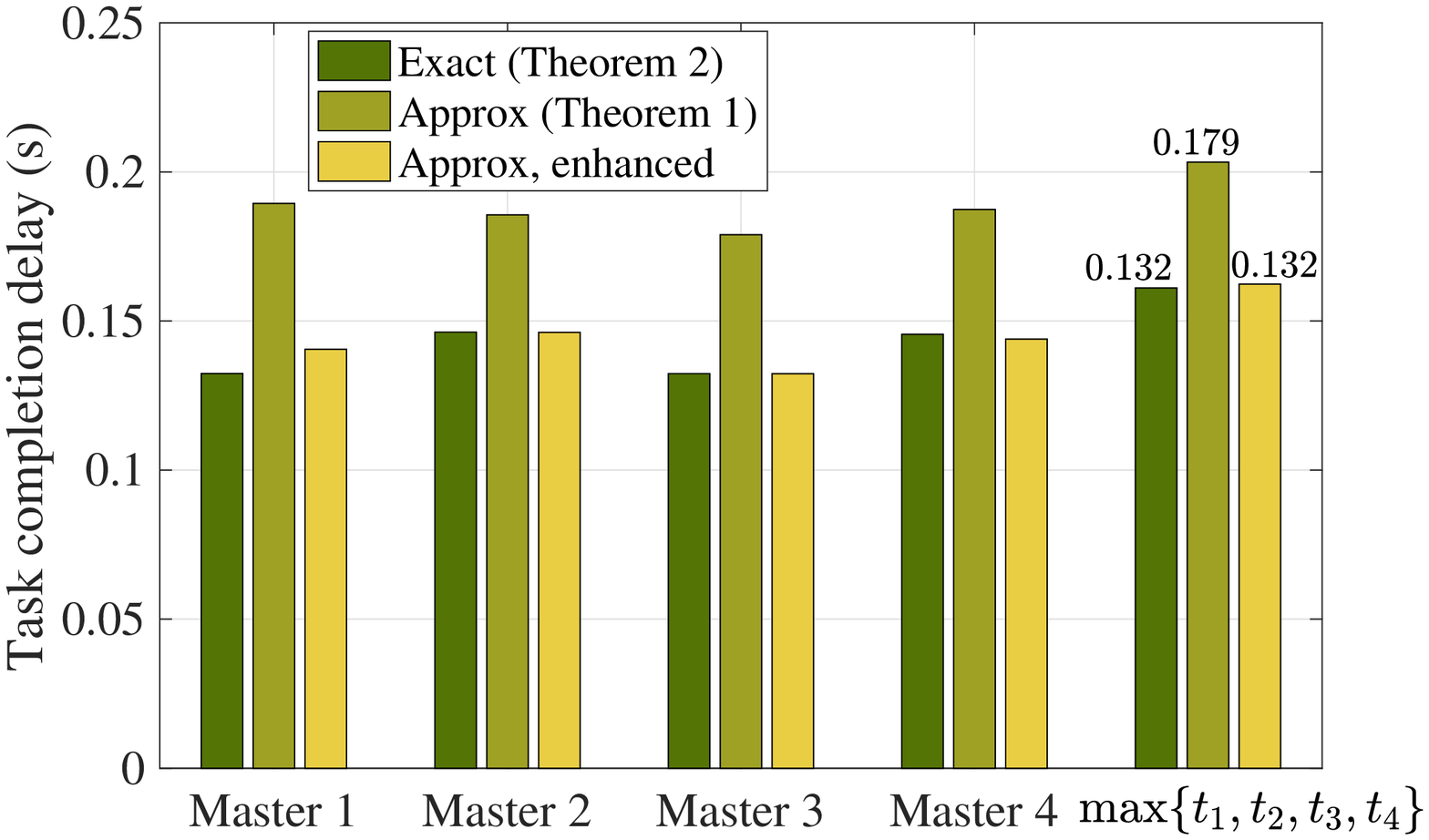}}
	\subfigure[CDF of task completion delay.]{\label{fig_val_m4n50_cdf}			
		\includegraphics[width=0.42\textwidth]{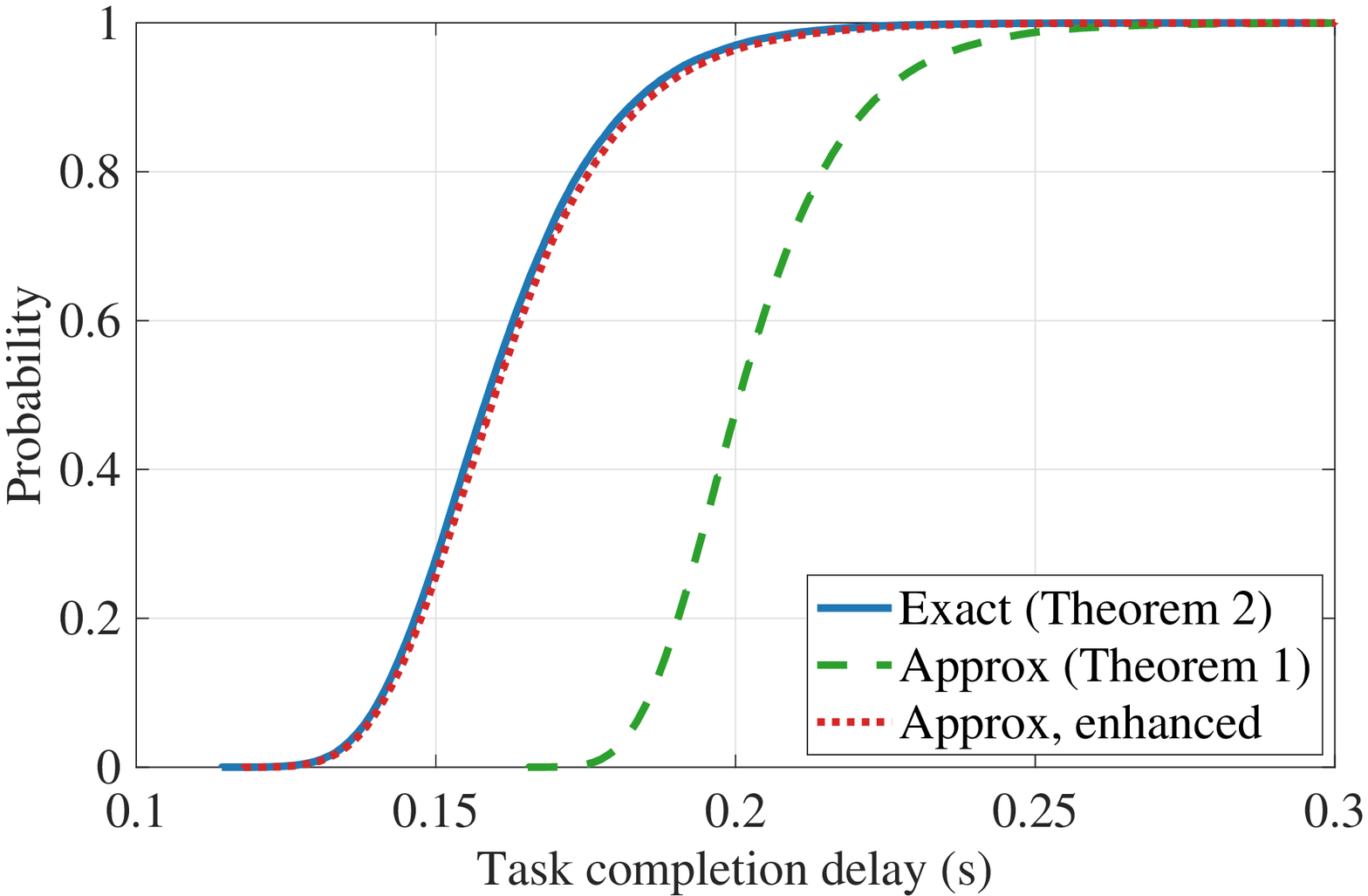}}
	\caption{Validation of the Markov's inequality based approximation in the $4$-master, $50$-worker case.}	
	\vspace{-1mm}\label{fig_val_m4n50}
\end{figure}

We first show in Fig. \ref{fig_val_m2n5} and Fig. \ref{fig_val_m4n50} that, the task completion delay achieved by solving the Markov's inequality based approximation problem is similar to the original problem. To achieve this verification, we consider the computation delay dominant case, where the optimal load allocation to the original problem $\mathcal{P}3$ (\emph{`Exact'}) can be derived from Theorem \ref{opt_load}, and the approximate load allocation (\emph{`Approx'}) is derived from Theorem \ref{markov_load}. Based on the two load allocation results, iterated greedy algorithm, i.e., Algorithm \ref{algo_iter}, is adopted to assign the workers in a dedicated manner. Corresponding to the SCA-enhanced load allocation in Section \ref{sub-SCA}, we further use Theorem \ref{opt_load} to improve the performance after obtaining the worker assignment based on the approximate load allocation, labeled as \emph{`Approx, enhanced'}.

We consider two scenarios with different scales. In the small-scale scenario, there are $M=2$ masters and $N=5$ workers. For each worker, the shift parameter $a_{m,n}$ of the computation delay distribution is randomly selected from $\{0.2,0.25,0.3\} ~\mathrm{ms}$, while for each master, $a_{m,0}\in\{0.4,0.5\} ~\mathrm{ms}$. The rate parameter is $u_{m,n}=\frac{1}{a_{m,n}}, \forall m,n$, and the load of the original task is set to $L_m=10^4, \forall m$ \cite{codedhet}. In the large-scale scenario, there are $M=4$ masters and $N=50$ workers. Parameter $a_{m,n}$ is randomly chosen from $[0.05,0.5] ~\mathrm{ms}$, while $u_{m,n}$ and $L_m$ remain the same. After deriving the load allocation and worker assignment from the corresponding theorems and algorithms, we run Monte Carlo realizations for $10^6$ times and present the average value and the CDF of the empirical task completion delay. 

%Note that, what we count in the simulation is the actual delay that tasks are indeed completed, concerning constraint \eqref{ori_cons_exp1} of the original problem $\mathcal{P}1$ with probability $\rho_s=1$, since the actual task completion delay is more sensible in practice.

Fig. \ref{fig_val_m2n5} and Fig. \ref{fig_val_m4n50} show the validation results under the small-scale and large-scale scenarios, respectively. In each histogram, the first $M$ groups of bars show the average task completion delay of each master under different solutions. The last group of bars show the average delay of all the tasks, which is what we aim to minimize in $\mathcal{P}2$, obtained by taking the maximum delay among $M$ masters in each Monte Carlo realization and then taking the average. Overall, the gap between the Markov's inequality based approximate solution and the optimal solution is acceptable, while the enhanced approximate solution has almost the same performance as the optimal one, in terms of both the average and the CDF of task completion delay under both scenarios.
We can also see from Fig. \ref{fig_val_m2n5_bar} that, the approximate solution can achieve even lower average delay when the number of workers is small. This is because, the approximate problem $\mathcal{P}4$ provides a tighter constraint and thus increases the redundancy of load, making the system more robust to stragglers in some cases.

\subsection{Performance of the Proposed Dedicated and Fractional Worker Assignment Algorithms }

Now we take the communication delay into account and evaluate the proposed algorithms. The simulation settings for both small-scale and large-scale scenarios remain the same as the previous subsection, while the communication rate parameter of each worker is set to $\gamma_{m,n}=2u_{m,n}, \forall m,n$. For the SCA algorithm, the step-size decreasing ratio is set to $\alpha=0.995$. We compare the delay performance of the proposed algorithms with the following benchmarks:

\emph{1) Uncoded computation with uniform worker assignment}: Each master is assigned an equal number of $\frac{N}{M}$ workers, and $\A_m$ is equally partitioned into $\frac{N}{M}$ sub-matrices without coding.
%, each with $\frac{L_mM}{N}$ rows.	

\emph{2) Coded computation with uniform worker assignment}: Each master is assigned an equal number of $\frac{N}{M}$ workers, and the load allocation is given by Theorem \ref{opt_load}. This benchmark can be regarded as the scheme presented in \cite{codedhet}, where only the computation delay is considered under a single master scenario.

\emph{3) Brute-force search for optimal fractional worker assignment}: The optimal benchmark is obtained by traversing all possible $k_{m,n}$ and $b_{m,n}$ at a step-size of $0.01$. SCA-enhanced load allocation is further implemented after getting the optimal fractional worker assignment. Note that, as the brute-force search is with extremely high complexity, we can only provide this result in the small-scale scenario.

\begin{figure}[t]
	\centering
	\subfigcapskip=-1mm
	\subfigure[$M=2$ masters, $N=5$ workers.]{\label{fig_2-5_bar}			
		\includegraphics[width=0.5\textwidth]{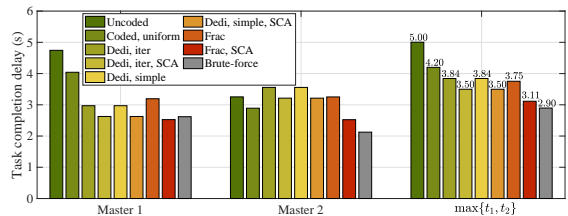}}\\
	%	\vspace{-3mm}\hspace{1mm}
	\subfigure[$M=4$ masters, $N=50$ workers.]{\label{fig_4-50_bar}			
		\includegraphics[width=0.5\textwidth]{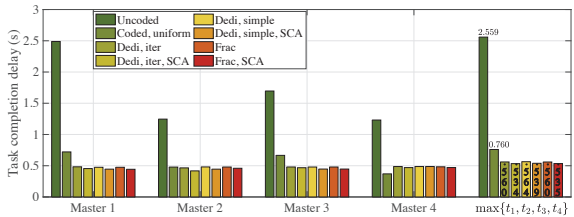}}
	%	\vspace{-3mm}
	\caption{Average task completion delay of the proposed algorithms and benchmarks.}	
	\vspace{-5mm}\label{fig_2-5}
\end{figure}

The average task completion delay in the two scenarios are shown in Fig. \ref{fig_2-5}. We use `Dedi, iter', `Dedi, simple' and `Frac' to represent the worker assignment results from Algorithms \ref{algo_iter}, \ref{algo_simple} and \ref{algo_frac}, respectively. The legend with `SCA' indicates that SCA-enhanced load allocation is further implemented. As shown in Fig. \ref{fig_2-5_bar}, in the small-scale scenario, our proposed algorithms outperform the uncoded and coded benchmarks by balancing the worker assignment, while the fractional assignment is slightly better than the dedicated one.
With SCA enhancement, the average delay can be decreased by $8.85\%$ under dedicated worker assignment, while the delay can be substantially decreased by $17.1\%$ with fractional assignment. We can also see that, the delay performance of the SCA-enhanced fractional assignment is close-to-optimal.
As shown in Fig. \ref{fig_4-50_bar}, in the large-scale scenario, iterated greedy algorithm can seek a better assignment compared to the simple greedy algorithm under the dedicated case. On the other hand, fractional assignment achieves the same performance as iterated greedy, since dedicated algorithm can already balance the worker assignment when the number of workers is large. With SCA-enhancement, the delay performance can be further decreased by over $4.4\%$, but we should also be aware that the complexity of the SCA algorithm is high under the large-scale scenario. Compared to the uncoded and coded benchmarks, up to $79\%$ and $30\%$ delay reduction can be achieved by the proposed algorithm, respectively.

\begin{figure}[t]
	\centering
	\subfigcapskip=-1mm
	\subfigure[$M=2$ masters, $N=5$ workers.]{\label{fig_2-5_cdf}			
		\includegraphics[width=0.5\textwidth]{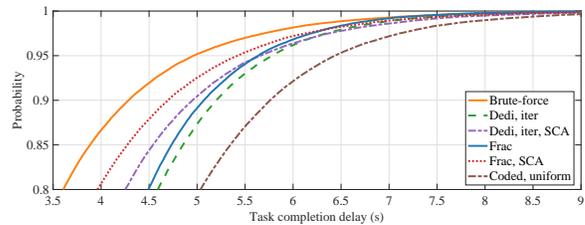}}\\
%	\vspace{-3mm}\hspace{1mm}
	\subfigure[$M=4$ masters, $N=50$ workers.]{\label{fig_4-50_cdf}			
		\includegraphics[width=0.5\textwidth]{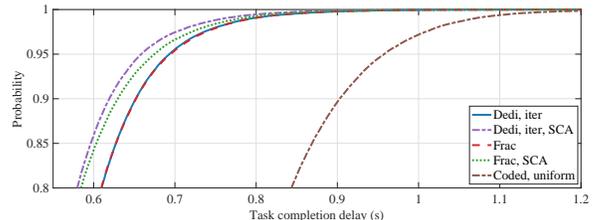}}
%	\vspace{-3mm}
	\caption{CDF of the task completion delay of the proposed algorithms and benchmarks.}	
	\vspace{-5mm}\label{fig_4-50}
\end{figure}

In order to observe the delay performance of the initial optimization problem $\mathcal{P}1$, we further plot the CDF of the task completion delay in Fig. \ref{fig_4-50}. Given the probability threshold $\rho_s$, we can obtain the corresponding delay from the x-axis, such that constraint \eqref{ori_cons_exp1} is satisfied. This figure shows the tail distribution of the task completion delay, and can reflect the robustness of the system under different algorithms. As shown in Fig. \ref{fig_4-50_cdf}, given  $\rho_s=0.95$, the delays achieved by the SCA-enhanced dedicated assignment, dedicated assignment and the coded benchmark are $0.658\mathrm{s}$, $0.694\mathrm{s}$ and $0.957\mathrm{s}$, respectively. That is, over $30\%$ delay reduction can be achieved by the proposed algorithm compared to the coded benchmark. We can also see that, a good solution to the approximation problem $\mathcal{P}2$ also leads to a good delay performance for the original problem $\mathcal{P}1$ in general, and thus solving $\mathcal{P}2$ is reasonable.

The impact of communication rate on the average task completion delay and the load allocation is investigated in Fig. \ref{fig_gamma_ratio_m4n50}, by varying $\frac{\gamma_{m,n}}{u_{m,n}}$ while fixing $u_{m,n}$. As shown in Fig. \ref{fig_gamma_m4n50}, when $\frac{\gamma_{m,n}}{u_{m,n}}$ is small, the communication rate between each master and worker is low, and thus the average task completion delay is high. Meanwhile, the proposed dedicated and fractional worker assignment algorithms always achieve significantly lower delay compared to the benchmarks. Fig. \ref{fig_gamma_lm0_m4n50} plots the ratio of load allocated to the master itself to the total load, i.e., $\frac{l_{m,0}}{\sum_{n\in\mathcal{N}'}l_{m,n}}$. As the two benchmarks do not take the communication delay into account, the ratio remains the same over different communication rates. With the proposed algorithms, this ratio decreases as $\frac{\gamma_{m,n}}{u_{m,n}}$ increases, since more computation load is allocated to the workers when communication is faster.
\begin{figure}[t]
	\centering
	\subfigcapskip=-1mm
%	\hspace{-3mm}
	\subfigure[Average task completion delay.]{\label{fig_gamma_m4n50}			
		\includegraphics[width=0.4\textwidth]{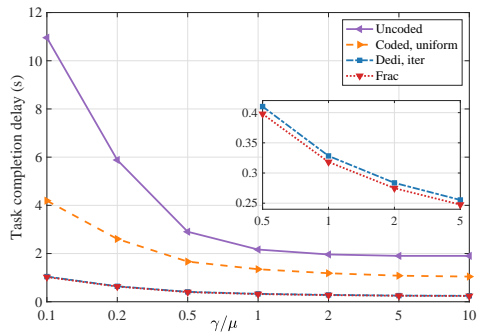}}
	%	\hspace{-6mm}
	\subfigure[The ratio of local processing load to total load.]{\label{fig_gamma_lm0_m4n50}			
		\includegraphics[width=0.4\textwidth]{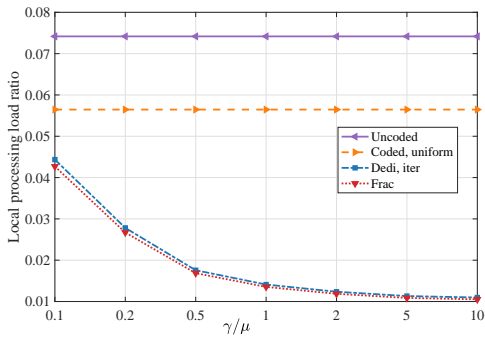}}
	%	\vspace{-3mm}
	\caption{Performance evaluation under different communication rates in the $4$-master, $50$-worker case.}	
	\vspace{-3mm}\label{fig_gamma_ratio_m4n50}
\end{figure}

\subsection{Delay Sampling on Amazon EC2 Instances and the Corresponding Algorithm Evaluation}

We further evaluate the delay performance of the proposed algorithms based on measured computation delays on the commercial compute platform Amazon EC2. To do so, we sample the computation delay on two types of Amazon EC2 instances called t2.micro and c5.large, by feeding each instance a $10^6$-dimension, float-number vector multiplication task for $10^6$ times. We plot the cumulative frequency distribution of the sampled computation delay in Fig. \ref{EC2_test}, and then fit the data with shifted exponential distribution. For the t2.micro instance, the shift parameter is $a=1.36~\mathrm{ms}$, and the rate parameter is $u=4.976~\mathrm{ms}^{-1}$. For c5.large instance, $a=0.97~\mathrm{ms}$ and $u=19.29~\mathrm{ms}^{-1}$. We can see that, in terms of the computation capability, the c5.large instance is more powerful than the t2.micro instance, and the fitting of the shifted exponential distribution is accurate.

%Simulation results based on real data computation delay, sampled from Amazon EC2. We first sample the CDF of computation delay by feeding the t2.micro/C5 instance a $10^8$-dimension vector multiplication task for $10^6$ times, where each entry is a float number within $[0,1]$. Then we plot the CDF and fit the curve with shifted exponential distribution. We use the fitted parameters to obtain the worker and load allocation, and use the real data to carry out the Monte Carlo simulation. 
%The fitted parameters are: t2.micro $a=0.136$, $u=49.76$, C5 $a=0.097$, $u=192.9$, no communication delay.

%\begin{figure}[!t]
%	\centering
%	\includegraphics[width=0.9\textwidth]{Figs//ec2_m2n5_1.eps}\vspace{-5mm}
%	\caption{Average task completion delay under the $2$-master, $5$-worker case.}	\vspace{-5mm}\label{t_amazon_m2n5}
%\end{figure}	

Finally, we use the measured data to evaluate the proposed algorithms. We consider a computation delay dominant scenario with $4$ masters and $50$ workers. All the masters and $40$ workers are considered as t2.micro instances, while the remaining $10$ workers are c5.large instances. We use the fitted distribution for load allocation and worker assignment, and then use the measured data to simulate the average task completion delay with the Monte Carlo method. As shown in Fig. \ref{t_amazon_m4n50}, the proposed dedicated and fractional worker assignment algorithms still outperform the uncoded and coded benchmarks, with up to $82\%$ and $30\%$ delay reductions, respectively. Comparing the two dedicated assignment algorithms, the iterated greedy algorithm achieves a much lower delay under this practical scenario. Meanwhile, fractional assignment can slightly decrease the average task completion delay compared with the iterated dedicated assignment.

Summarizing all the simulation results, we remark that, in the small-scale scenario where the number of master and worker nodes is small, SCA-enhanced fractional assignment is the best algorithm with great advantage over other alternatives. On the other hand, in the large-scale scenario, dedicated assignment by the iterated greedy algorithm is satisfactory when considering the delay performance together with the complexity of the algorithm and the network topology.

\begin{figure}[!t]
	\centering	
	%	\vspace{-5mm}
	\subfigcapskip=-1mm
	\subfigure[t2.micro instance.]{\label{CDF_micro}			
		\includegraphics[width=0.48\textwidth]{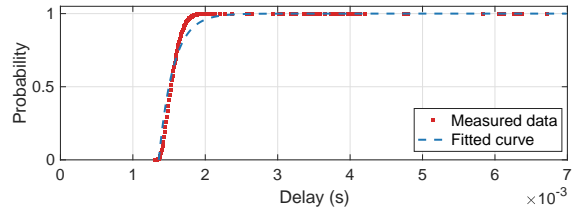}}	\\
	%	\vspace{-3mm}	
	%	\hspace{1mm}
	\subfigure[c5.large instance.]{\label{CDF_large}			
		\includegraphics[width=0.48\textwidth]{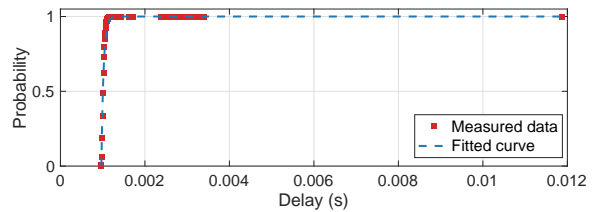}}
	%	\vspace{-3mm}
	\caption{The cumulative frequency distribution of the measured delay and its fitted curve based on shifted exponential distribution, on different types of Amazon EC2 instances.}
	\vspace{-1mm}
	\label{EC2_test}
\end{figure} 

\begin{figure}[!t]
%	\vspace{-4mm}
	\centering
	\includegraphics[width=0.5\textwidth]{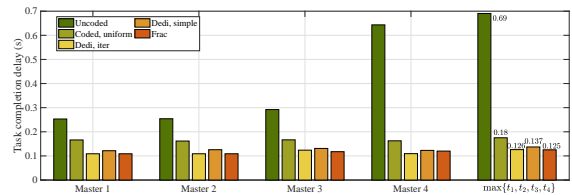}\vspace{-2mm}
	\caption{Average task completion delay under the $4$-master, $50$-worker case.}	
	\vspace{-1mm}\label{t_amazon_m4n50}
\end{figure}
	
\begin{comment}

%\begin{figure}[t]  cdf_m2n5
%	\centering
%	\subfigure[$M=2$, $N=5$.]{\label{fig_gamma_mu_m2n5}			
%		\includegraphics[width=0.45\textwidth]{Figs//gamma_mu_m2n5.eps}}
%	\subfigure[$M=4$, $N=50$.]{\label{fig_gamma_mu_m4n50}			
%		\includegraphics[width=0.45\textwidth]{Figs//gamma_mu_m4n50.eps}}
%	\caption{Fix $\frac{1}{\gamma_{m,n}}+\frac{1}{u_{m,n}}=2a_{m,n}$, and vary $\frac{\gamma_{m,n}}{u_{m,n}}$.
%		\textcolor{red}{Convexity???}}	\vspace{-5mm}\label{fig_gamma_mu_ratio}
%\end{figure}

%\begin{figure}[t]
%	\centering
%	\hspace{-3mm}
%	\subfigure[Task completion delay.]{\label{fig_gamma_m2n5}			
%		\includegraphics[width=0.49\textwidth]{Figs//gamma_m2n5.eps}}
%	\hspace{-6mm}
%	\subfigure[Local processing load.]{\label{fig_gamma_lm0_m2n5}			
%		\includegraphics[width=0.49\textwidth]{Figs//gamma_lm0_m2n5.eps}}
%	\caption{$M=2$, $N=5$. Fix $a_{m,n}$ and $u_{m,n}$, and vary $\frac{\gamma_{m,n}}{u_{m,n}}$ (which is to vary the communication capability $\frac{1}{\gamma_{m,n}}$).}	\vspace{-5mm}\label{fig_gamma_ratio}
%\end{figure}

\end{comment}	

\section{Conclusions} \label{con}
We have investigated a joint worker assignment, resource allocation, and load allocation problem in an MDS-coded distributed computing scenario with multiple masters and heterogeneous workers, aiming to minimize the communication and computation delay of tasks. Dedicated and fractional worker assignment and load allocation algorithms have been proposed, employing Markov's inequality-based approximation, Karush-Kuhn-Tucker conditions, and SCA techniques for the analysis and optimization of these algorithms. 
Simulations under various settings have shown that the proposed algorithms can significantly reduce the task completion delay compared to the benchmark algorithms, while the SCA-enhanced fractional assignment algorithm can achieve close-to-optimal delay performance when the number of master and worker nodes is small.
Considering measured data on Amazon EC2 platform for delay evaluation, we have shown that about $82\%$ and $30\%$ delay reductions can be achieved by the proposed algorithms compared to the uncoded and coded benchmarks, respectively. 
We have observed that SCA-enhanced fractional assignment significantly outperforms the other proposed algorithms under small-scale scenarios. Meanwhile, the dedicated policy with iterated greedy assignment can be a practical alternative for large-scale scenarios, when the delay performance, resultant communication network topology, and algorithm complexity are jointly taken into account.

As future directions, multi-message communication schemes \cite{Ozfatura2020straggler} as well as the costs of encoding and decoding can be further incorporated into the current optimization framework.

\appendices{}

\section{Proof of Theorem 1} \label{a1}
For $x>0,~y>0$, $f(x,y)=\frac{x^2}{y}$ is convex. Therefore, problem $\mathcal{P}4$ is a convex optimization problem.
The Lagrangian of  $\mathcal{P}4$ is given by
\begin{align}
&\mathcal{L}(\l_m,t_m, \lambda_m) \!=\! t_m \!\!+ \!\!
\lambda_m\! \left[\!L_m\!-\!\sum_{n\in\Omega_m'}\!\left(\!l_{m,n}\!-\!\frac{\theta_{m,n} l_{m,n}^2}{t_m}\!\right) \!\right], \nonumber
\end{align}
where $\lambda_m\geq 0$ is the Lagrange multiplier associated with \eqref{con_dedi_approx}.

The partial derivatives of $\mathcal{L}(\l_m,t_m, \lambda_m)$ can be derived as
\begin{subequations}
	\begin{align}
	&\frac{\partial \mathcal{L}}{\partial l_{m,n}}=-\lambda_m+\lambda_m\theta_{m,n} \frac{2l_{m,n}}{t_m},  \\
	~&\frac{\partial \mathcal{L}}{\partial t_m}=1-\lambda_m\sum_{n\in\Omega_m'}\frac{\theta_{m,n} l_{m,n}^2}{t_m^2}. 
	\end{align}
\end{subequations}

The Karush-Kuhn-Tucker (KKT) conditions are written as follows:
\begin{subequations}
	\begin{align}
	&\frac{\partial \mathcal{L}}{\partial l_{m,n}^*}=0,~\forall n\in \Omega_m',~~ 
	\frac{\partial \mathcal{L}}{\partial t_m^*}=0, \\
	&\lambda_m^*\left[L_m-\sum_{n\in\Omega_m'}\left(l_{m,n}^*-\frac{\theta_{m,n} {l_{m,n}^*}^2}{t_m^*}\right) \right]=0,\\
	&\lambda_m^*\geq 0,~~l_{m,n}^*\geq 0.
	\end{align}
\end{subequations}

By solving the KKT conditions, we get the optimal load allocation and task completion delay to $\mathcal{P}4$, as shown in Theorem 1.
%\begin{subequations}
%	\begin{align}
%	&t_m^*=\frac{L_m}{\sum_{n\in\Omega_m'} \frac{1}{4\theta_{m,n}}}, ~l_{m,n}^*=\frac{L_m}{\theta_{m,n}\sum_{n\in\Omega_m'} \frac{1}{2\theta_{m,n}}},~\forall n\in \Omega_m'. 
%	\end{align}
%\end{subequations}

\section{Proof of Theorem 2} \label{a2}
When the computation delay dominates the total delay, the optimization problem $\mathcal{P}3$ is given by
\begin{subequations}
	\begin{align}
	\mathcal{P}3(1) \!:\! \min_{\{\l_m,~t_m\}} &~~~~t_m  \nonumber\\
	\text{s.t.} ~~~&L_m \!\leq\!\!\!\!\sum_{n\in \Omega_m'}\!\! l_{m,n}\! \left(\!  1-\! e^{-\frac{u_{m,n}}{l_{m,n}}\left(t_m-a_{m,n}l_{m,n}\! \right)} \! \right),  \label{con_dedi_cp} \nonumber\\
	&~l_{m,n}\geq 0 , ~\forall n\in \Omega_m'. \nonumber
	\end{align}
\end{subequations}

We first prove that $\mathcal{P}3(1)$ is a convex optimization problem.
Let $f(x,t)=-x \left( 1-e^{-\frac{u}{x}(t-ax)}\right)$,
%\begin{align}
%f(x,t)=-x \left( 1-e^{-\frac{u}{x}(t-ax)}\right),  \label{fxt}
%\end{align}
with variables $x>0$, $~t\geq ax$, and parameters $u>0$, $a>0$.
The Hessian matrix of $f(x,t)$ is:
\begin{align}
\H= \left[
\begin{matrix}
\frac{\partial^2 f }{\partial x^2 }  &\frac{\partial^2 f }{\partial x \partial t }\\
\frac{\partial^2 f }{\partial t \partial x} &\frac{\partial^2 f }{\partial  t^2 }
\end{matrix}
\right]=e^{-\frac{u}{x}(t-ax)} \left[
\begin{matrix}
\frac{u^2t^2}{x^3}  &-\frac{u^2t}{x^2}\\
-\frac{u^2t}{x^2}&\frac{u^2}{x}
\end{matrix}
\right].
\end{align}
The eigenvalues of $\H$ are $0$ and $\frac{u^2(x^2+t^2)}{x^3}>0$. Thus $\H \succeq 0$, and $f(x,t)$ is convex.
As the summation of convex functions are still convex, \eqref{con_dedi_cp} is convex. Therefore, $\mathcal{P}3(1)$ is convex.

The Lagrangian  is given by
\begin{align}
&\mathcal{L}(\l_m,t_m, \lambda_m) = t_m+\lambda_m\left(L_m-\mathbb{E}[X_m(t_m)]\right) \nonumber\\
&=\!t_m \!\!+\!\lambda_m\!\!\! \left[ \!L_m\!-\!\!\!\!\sum_{n\in \Omega_m'}\!\!\! l_{m,n}\! \left(\! 1-e^{-\frac{u_{m,n}}{l_{m,n}}\left(t_m-a_{m,n}l_{m,n}\right)}\! \right) \!\right]. \label{p3_lagrange}
\end{align}

The partial derivatives of $\mathcal{L}$ can be derived as
\begin{subequations}
\begin{align}
\frac{\partial \mathcal{L}}{\partial l_{m,n}}\!&=\!
\lambda_m \!\left[\left(1+\frac{u_{m,n}t_m}{l_{m,n}}\right)e^{-\frac{u_{m,n}}{l_{m,n}}\left(t_m-a_{m,n}l_{m,n}\!\right)} \!-\!1  \right],  \label{partial_l} \\
\frac{\partial \mathcal{L}}{\partial t_m}&=1-\lambda_m\sum_{n\in \Omega_m'}u_{m,n}e^{-\frac{u_{m,n}}{l_{m,n}}\left(t_m-a_{m,n}l_{m,n}\right)}.  \label{partial_t}
\end{align}
\end{subequations}

The optimal solution $(\l^*_m,t^*_m, \lambda^*_m)$ needs to satisfy the following KKT conditions
\begin{subequations}
	\begin{align}
	&\frac{\partial \mathcal{L}}{\partial l^*_{m,n}}=0, ~\forall n\in \Omega_m' ,
	~~\frac{\partial \mathcal{L}}{\partial t^*_m}=0,  \label{kkt_2} \\ 
	&\lambda^*_m \!\! \left[ \!L_m\!\!-\!\!\!\sum_{n\in \Omega_m'} \!\!l_{m,n}^* \!\left( \! 1\!-\!e^{-\frac{u_{m,n}}{l_{m,n}^*}\left(t_m^*\!-a_{m,n}l_{m,n}^*\!\right)}\! \right)\!\right]\!=\!0, \label{kkt_3}  \\
	&\lambda^*_m\geq 0, ~l^*_{m,n}\geq 0.  \label{kkt_4}
	\end{align}
\end{subequations}

By jointly considering \eqref{partial_t} and \eqref{kkt_2}, we get $\lambda_m^*>0$. 
Substituting \eqref{partial_l} into $\frac{\partial \mathcal{L}}{\partial l_{m,n}^*}=0$ yields
\begin{align}
&	\left(1+\frac{t_m^*u_{m,n}}{l_{m,n}^*}\right)e^{u_{m,n}\left(a_{m,n}-\frac{t_m^*}{l_{m,n}^*}\right)}-1=0, \nonumber\\
&
-\left(1+\frac{t_m^*u_{m,n}}{l_{m,n}^*}\right)e^{-\left(1+\frac{t_m^*u_{m,n}}{l_{m,n}^*}\right)}=-e^{-u_{m,n}a_{m,n}-1}, \nonumber
%&   1+\frac{t_m^*u_{m,n}}{l_{m,n}^*}=-\mathcal{W}_{-1}(-e^{-u_{m,n}a_{m,n}-1} ),
\end{align}
Let $\mathcal{W}_{-1}(x)$ be the lower branch of Lambert W function, where $x\leq -1$ and $\mathcal{W}_{-1}(xe^x)=x$. Then we have
\begin{align} \label{opt_t_over_l}
\frac{t^*_m}{l^*_{m,n}}=\frac{-\mathcal{W}_{-1}(-e^{-u_{m,n}a_{m,n}-1} )-1}{u_{m,n}}\triangleq \phi_{m,n}.
\end{align}
Substituting \eqref{opt_t_over_l} into \eqref{kkt_3}:
\begin{align}
&L_m-\sum_{n\in \Omega_m} \frac{t^*_m}{\phi_{m,n}}\left(1-\frac{1}{1+u_{m,n}\phi_{m,n}}\right)=0.
\end{align}
Then, $t_m^* $ and $l^*_{m,n}$ can be derived, as shown in Theorem \ref{opt_load}.

\section{Proof of Theorem 3} \label{a3}

Given any $\k$ and $\b$, the Lagrangian of  $\mathcal{P}6$ is given by
\begin{align}
&\mathcal{L}(\l,t, \lambda_m) \nonumber\\
&= t+
\lambda_m\left[L_m-\sum_{n\in\mathcal{N}'}\left(l_{m,n}-\frac{\theta_{m,n} l_{m,n}^2}{t}\right) \right], \forall m. \label{p6_lagrange}
\end{align}

For the non-convex optimization problem, the optimal solution $\{\l^*,t^*\}$ must satisfy the KKT conditions. By solving 
\begin{subequations}
	\begin{align}
	&\frac{\partial \mathcal{L}}{\partial l_{m,n}^*}=-\lambda_m^*+\lambda_m^*\theta_{m,n} \frac{2l_{m,n}^*}{t^*}=0,   \nonumber\\
	&\frac{\partial \mathcal{L}}{\partial t^*}=1-\lambda_m^*\sum_{n\in\mathcal{N}'}\frac{\theta_{m,n} (l_{m,n}^*)^2}{(t^*)^2}=0, \nonumber\\
	&\lambda_m^*\left[L_m-\sum_{n\in\mathcal{N}'}\left(l_{m,n}^*-\frac{\theta_{m,n} (l_{m,n}^*)^2}{t^*}\right) \right]=0,   \nonumber\\
	&\lambda_m^*\geq 0, ~l_{m,n}^*\geq 0,
	\end{align}
\end{subequations}
we derive the optimality condition, as shown in Theorem \ref{opt_cond}.

\end{document}